\documentclass[runningheads,envcountsect,envcountsame]{llncs}
\usepackage{orcidlink}
\usepackage{lmodern}
\usepackage[T1]{fontenc}

\newif\ifshort
\shortfalse

\usepackage{soul}
\usepackage{url}
\usepackage[utf8]{inputenc}

\usepackage{amsmath,amssymb,graphicx,dsfont,mathtools}
\usepackage{paralist}
\usepackage{booktabs}

\usepackage[justification=justified]{caption}

\usepackage{hyperref}
\hypersetup{%
  backref=true,
  pagebackref=true, 
  hypertexnames=true,
  urlcolor=blue,
  colorlinks=true,citecolor=green!60!black,linkcolor=red!60!black%
}
\usepackage[capitalize]{cleveref}
\usepackage{etoolbox}
\AtEndEnvironment{proof}{\qed}

\newtheorem{clm}{Claim}[section]

\crefname{table}{Table}{Tables}
\crefname{figure}{Figure}{Figures}
\crefname{theorem}{Theorem}{Theorems}
\crefname{corollary}{Corollary}{Corollaries}
\crefname{observation}{Observation}{Observations}
\crefname{lemma}{Lemma}{Lemmas}
\crefname{example}{Example}{Examples}
\crefname{reduction}{Reduction}{Reductions}
\crefname{construction}{Construction}{Constructions}
\crefname{subsection}{Subsection}{Subsections}
\crefname{section}{Section}{Sections}
\crefname{claim}{Claim}{Claims}
\crefname{clm}{Claim}{Claims}
\crefname{algorithm}{Algorithm}{Algorithms}
\crefname{algocf}{Algorithm}{Algorithms}

\crefname{definition}{Definition}{Definitions}

\usepackage[textsize=tiny,textwidth=1.5cm,linecolor=green!70!black, backgroundcolor=green!10, bordercolor=black,disable]{todonotes}

\usepackage[noend,ruled,linesnumbered]{algorithm2e} %

\SetAlFnt{\small}
\SetAlCapFnt{\small}
\SetAlCapNameFnt{\small}
\SetAlCapHSkip{0pt}
\IncMargin{-\parindent}

\SetKwInput{KwI}{Input}
\SetKwInput{KwO}{Output}
\SetKw{KwR}{return}
\SetKw{KwNo}{no}
\SetKw{KwOr}{or}
\SetKw{KwAnd}{and}

\SetCommentSty{mycommfont}
\SetKwComment{Comment}{$\triangleright$\ }{}

\usepackage[switch]{lineno}

\usepackage{tabularx}
\usepackage{multirow}
\usepackage{colortbl}
\usepackage{tikz}
\usetikzlibrary{matrix,shapes,arrows,calc,arrows,fit}
\usepackage{rotating}
\usepackage{pifont}

\usepackage{mdframed}

\usepackage[textsize=tiny,textwidth=1.5cm,linecolor=green!70!black, backgroundcolor=green!10, bordercolor=black,disable]{todonotes}

\newcommand{\mypara}[1]{
  \medskip

\noindent  {\textup{\textbf{#1}}}}

\newcommand{\myparaI}[1]{
	
	\smallskip
	\noindent {\it #1}
}

\usepackage{pict2e}
\newcommand{\opentriangle}{%
  \raisebox{0.2pt}{\makebox[0.77778em]{%
    \setlength{\unitlength}{0.6em}%
    \linethickness{0.4pt}%
    \begin{picture}(1,1)
    \polygon(0,0)(1,0)(1,1)
    \end{picture}%
  }}%
}
\newenvironment{proofsketch}[1]{\noindent\emph{Proof sketch.}\hspace{0.15cm}#1}{\hfill$\opentriangle$\medskip}

\usepackage[title]{appendix}
\usepackage{thm-restate}
\newcommand{\linkproof}[1]{%
    \ifshort
      $\star$%
    \else
      \hyperref[app:proof:#1]{$\star$}%
    \fi
}

\newcommand{\prooftoappendix}[3]{%
  \gappto{\appendixtext}{
    \subsection{Proof of \Cref{#1}}\label{app:proof:#1}
    #2
    \begin{proof}
    #3\end{proof}
  }
}

\newcommand{\appendixcorrectnessproofwithstatement}[4]{%
  #1  
  \gappto{\appendixtext}{
    \subsection{Correctness of the Construction in the Proof
      of \cref{#2}}\label{app:proof:#2}
    {#3}
    
    #4
  }
}
\newcommand{\appendixsection}[1]{%
  \gappto{\appendixtext}{
    \section{Additional material for \Cref{#1}}
    \label{app:#1}
  }
}

\newcommand{\p}{\ensuremath{\mathrm{P}}}
\newcommand{\np}{\ensuremath{\mathrm{NP}}}
\newcommand{\npc}{\ensuremath{\mathrm{NPc}}}

\newcommand{\CSR}{\textsc{CSR}}
\newcommand{\CSM}{\textsc{CSM}}
\newcommand{\RSAT}{\textsc{(2P,2N)-3SAT}}
\newcommand{\TSAT}{\textsc{3SAT}}

\newcommand{\nn}{\ensuremath{\hat{n}}}
\newcommand{\mm}{\ensuremath{\hat{m}}}
\newcommand{\xii}{\ensuremath{{x}_i}}
\newcommand{\yi}{\ensuremath{{y}_i}}
\newcommand{\nxi}{\ensuremath{\overline{x}_i}}
\newcommand{\nyi}{\ensuremath{\overline{y}_i}}
\newcommand{\lit}{\ensuremath{l}}
\newcommand{\tru}{\ensuremath{\mathsf{T}}}
\newcommand{\fal}{\ensuremath{\mathsf{F}}}
\newcommand{\zstar}{\ensuremath{z^*}}

\newcommand{\mydeg}[1]{\mathsf{d}_{#1}}
\def\symdiff{\Delta\,}

\usepackage[numbers]{natbib}

\newcommand{\decprob}[3]{
  \begin{center}%
    \begin{minipage}{0.9\linewidth}%
      \textsc{#1}\\[0.2ex]
      \textbf{Input:} #2\\[0.2ex]
      \textbf{Question:} #3
    \end{minipage}%
  \end{center} }

\newenvironment{claimproof}[1]{\par\noindent\textit{Proof.}\space#1}{\hfill $\lhd$\bigskip}

\title{Control and Bribery in {S}table {M}arriage and {S}table {R}oommates: A Complete Complexity Landscape}
\titlerunning{Control and Bribery in {S}table {M}arriage and {S}table {R}oommates}

 \author{Jiehua Chen\inst{1}\orcidlink{0000-0002-8163-1327}
 \and 
 Ildik\'{o} Schlotter\inst{2,3}\orcidlink{0000-0002-0114-8280}
 }

\authorrunning{Jiehua Chen and Ildikó Schlotter}

 \institute{
 TU Wien, Austria, \email{jiehua.chen@ac.tuwien.ac.at }
 \and
 ELTE Centre for Economic and Regional Studies, Hungary, 
 \email{schlotter.ildiko@krtk.elte.hu}
 \and
 Budapest University of Technology and Economics, Hungary
 }

\newcommand{\myemph}[1]{{\color{purple!50!black}\textit{#1}}}

\def\A{\mathcal{A}}
\def\G{\mathcal{G}}
\def\AddAg{\mathsf{AddAg}}
\def\DelAg{\mathsf{DelAg}}
\def\DelAcc{\mathsf{DelAcc}}

\def\Swap{\mathsf{Swap}}
\def\Reorder{\mathsf{Reorder}}
\def\MA{\mathsf{MA}}
\def\MP{\mathsf{MP}}
\def\SM{\mathsf{MS}}
\def\USM{\mathsf{USM}}
\def\ESM{\exists\mathsf{SM}}
\def\PSM{\exists\mathsf{PSM}}

\def\SR{\textsc{Stable Roommates}}

\begin{document}
\maketitle

\begin{abstract}
  We study control and bribery problems for stable matchings: a central authority (the \emph{controller}, resp., \emph{briber}) may add agents, delete agents, delete acceptable pairs, swap two adjacent agents in some agent's preference list, or arbitrarily reorder some agent's preference list, in an instance of \textsc{Stable Marriage} or \textsc{Stable Roommates}.
  We extend previous work on control and bribery in stable matchings by Boehmer et al.~\cite{boe-bre-hee-nie:j:control-bribery-in-SM}.
  We consider goals capturing individual and pair inclusion, stability, and uniqueness requirements: matching a designated agent ($\MA$), matching a designated pair ($\MP$), realizing a stable matching consistent with a given matching ($\SM$), making a given matching the unique stable matching ($\USM$), or guaranteeing that a stable (resp., perfect and stable) matching exists ($\ESM$/$\PSM$).
  We provide a unified complexity map for all non-trivial action--goal combinations in both settings, consolidating known results and extending the study to the roommates model, where stable matchings need not exist.
\end{abstract}

\section{Introduction}

\SR\ deals with the division of human or non-human agents into pairs,
where  agents have preference lists describing their desires about whom they want to be matched with.
More precisely, an instance of \textsc{Stable Roommates} consists of a set~$U$ of agents,
with each agent~$u \in U$ having a strict preference list over the subset
of the remaining agents that $u$ finds \myemph{acceptable}.
The goal is to decide whether there is a \myemph{stable matching}, that is, a set~$M$ of pairwise disjoint pairs of agents for which \myemph{no} blocking pair exists.
Herein, a \myemph{blocking pair} for~$M$ is a pair of distinct agents $u$ and $w$ who are not matched together by~$M$ such that
\begin{compactenum}[(i)]
  \item  $u$ is either unmatched by $M$ or prefers $w$ to its partner in~$M$
  and
  \item $w$ is either unmatched by $M$ or prefers $u$ to its partner in~$M$.
\end{compactenum}

Note that each \textsc{Stable Roommates} instance induces an \myemph{acceptability graph} where the vertices correspond to the agents and there is an edge between two vertices if both consider each other acceptable. 
A well-known special case of \textsc{Stable Roommates} called \textsc{Stable Marriage}
is the \myemph{bipartite} variant where agents are partitioned into two disjoint subsets such that each agent's preference list consists of agents of the opposite subset, and each agent can only be assigned a partner from that subset.
The acceptability graph of an instance of \textsc{Stable Marriage} is then bipartite.
Due to a celebrated result by Gale and Shapley~\cite{gal-sha:j:college-admissions-and-the-stability-of-marriage}, we know that in the marriage setting a stable matching always exists and can be found in time linear in the input size.
However, for the roommates setting, a stable matching is not guaranteed to exist and, more than two decades after the seminal paper by Gale and Shapley~\cite{gal-sha:j:college-admissions-and-the-stability-of-marriage}, 
Irving~\cite{irv:j:efficient-algorithm-for-stable-roommates-problem} provided a polynomial-time algorithm to decide whether a given instance of \textsc{Stable Roommates} admits a stable matching.

Bartholdi III et al.~\cite{bar-tov-tri:j:manipulating, bar-tov-tri:j:control} introduced a framework for considering questions of manipulation and control in the context of elections, and their approach was adopted by
Boehmer et al.~\cite{boe-bre-hee-nie:j:control-bribery-in-SM} to investigate issues of manipulation, control, and bribery in the context of \textsc{Stable Marriage}.
In their framework, an external agent (the \myemph{controller}, also called the \myemph{briber} when modifications affect preferences directly) aims to achieve a specific \myemph{goal} by applying as few \myemph{actions} as possible.

\begin{example}[\cite{ChenKNRSS-ijcai25}]
\label{exa:stable-matching}
Suppose that Anna~($a$), Belle~($b$), Chris~($c$), and David~($d$) are attending a dance class, and need to form opposite-sex couples.
Their preferences are as follows:%
\[
\nolinenumbers
a\colon c \succ d, \qquad
b\colon c \succ d, \qquad
c\colon a \succ b, \qquad
d\colon b \succ a.
\]
The only stable matching in this instance is $\{\{a,c\},\{b,d\}\}$.
However, Belle's friend, evil Eve, is among the teachers of the class, and her goal is to match Belle with her top-choice partner, Chris.
She considers three options to achieve her goal:
(1)~declaring that Chris is too short for Anna and hence cannot dance with her,
(2)~stepping on Anna's toes with her high heels, thereby sending her off to ER, or
(3)~inviting her attractive friend, Frodo ($f$), whom Anna prefers to Chris.
She decides on option (3), and obtains
the following instance:
\[
\nolinenumbers
a\colon f \succ c \succ d, \quad
b\colon c \succ f \succ d, \quad
c\colon a \succ b, \quad
d\colon b \succ a, \quad
f\colon b \succ a.
\]
Yet, Eve is not entirely satisfied, as now there are two stable matchings, namely $M_1=\{\{a,f\},\{b,c\}\}$ and $M_2=\{\{a,c\},\{b,f\}\}$. Hence, she declares that Frodo is too tall for Belle, ensuring that $M_1$ becomes the only stable matching.
\end{example}

\Cref{exa:stable-matching} illustrates a few of the manipulative actions and goals examined in the literature on control and bribery for stable matchings.
However, despite the apparent simplicity of the underlying space of action--goal--setting combinations, \myemph{previous work has so far only addressed disjoint fragments of it}:
Boehmer et al.~\cite{boe-bre-hee-nie:j:control-bribery-in-SM} consider five manipulative actions, but only three goals, and exclusively in the bipartite \textsc{Stable Marriage} setting;
B{\'e}rczi et al.~\cite{ber-csa-kir:j:manipulation-SR} consider two further (existence-based) goals, but only for the action of deleting agents and in \textsc{Stable Roommates};
and individual action--goal cases have appeared under different names in work on almost-stable~\cite{abr-bir-man:c:almost-stable-roommates,bir-man-mit:j:size-vs-stability}, maximum-cardinality stable~\cite{tan:j:maximum-stable-matching,tan:j:stable-partitions}, and covering-constraint stable~\cite{mni-sch:j:SM-covering-constraints} matchings.
The result is a \emph{patchwork} of partial results, with no unified picture of which combinations are tractable, which are hard, and how the marriage and roommates settings relate.

In this paper, we close this gap.
We \myemph{unify} these strands into a single coherent framework, considering \emph{all} five manipulative actions of Boehmer et al.\ under \emph{six} goals, %
in \emph{both} the marriage and the non-bipartite roommates settings.
Following the terminology of Boehmer et al., we organize the five actions into two complementary paradigms.
\looseness=-1
\myparaI{Control actions.}
In many matching markets, an organizer, platform, or regulator can change \emph{participation} and \emph{feasibility} without rewriting preferences: participants may enter or leave the market, and certain pairs may be ruled out (e.g., due to conflicts of interest or policy restrictions); moreover, unlike direct preference manipulation, such interventions can often be justified by external rules.
We therefore consider the following control actions, showcased in \cref{exa:stable-matching}:
\begin{compactitem}
\item $\AddAg$: adding agents (e.g., inviting Frodo),
\item $\DelAg$: deleting agents (e.g., removing Anna from the class), and
\item $\DelAcc$: deleting acceptability (e.g., declaring constraints on who cannot dance with whom).
\end{compactitem}
As in classic  control models in computational social choice, we minimize the number of actions, capturing intervention cost and market disruption.

\myparaI{Bribery actions.}
A complementary type of intervention models agents who can be \emph{bribed}, i.e., convinced to misreport their preferences; it is captured by the following two actions of Boehmer et al.~\cite{boe-bre-hee-nie:j:control-bribery-in-SM}:
\begin{compactitem}
\item $\Swap$: swapping two adjacent agents in some agent's preference list (modeling a briber who pays per ``small'' modification),
\item $\Reorder$: arbitrarily permuting some agent's preference list (modeling a briber who pays per fully ``bought'' agent).
\end{compactitem}
As with control actions, we minimize the number of bribery operations.

\myparaI{Goals.}
The goals we study correspond to common individual and pair inclusion, stability, and uniqueness requirements. Formally, given an instance of \SR, we consider the goals:
\begin{compactitem}
\item $\MA$: a given agent is matched in some stable matching,
\item $\MP$: a given pair is contained in some stable matching,
\item $\SM$: there exists a stable matching consistent with a given matching~$M$ (i.e., using only pairs from~$M$),
\item $\USM$: a given matching~$M$ is the unique stable matching,
\item $\ESM$: a stable matching exists, or
\item $\PSM$: a perfect and stable matching exists.
\end{compactitem}
Among these, $\MP$, $\SM$, and~$\USM$ were studied (under different names) by Boeh\-mer et al.~\cite{boe-bre-hee-nie:j:control-bribery-in-SM}, $\ESM$ and~$\PSM$ by B{\'e}rczi et al.~\cite{ber-csa-kir:j:manipulation-SR}, and $\MA$ has been studied in a slightly different context by Mnich and Schlotter~\cite{mni-sch:j:SM-covering-constraints}.

\myparaI{Our Problems.}
For each action~$\mathcal{A} \in \{\AddAg,\DelAg,\DelAcc,\Swap,\Reorder\}$ and each goal~$\mathcal{G} \in \{\MA,\MP,\SM,\USM,\ESM,\PSM\}$, the \textsc{Control for Stable Roommates}-$\A$-$\G$ (abbreviated as \CSR-$\A$-$\G$) problem asks for a minimum number of actions of type~$\A$ necessary to achieve the goal~$\G$ in a given \SR\ instance.
When restricted to the bipartite marriage setting, we call these problems \textsc{Control for Stable Marriage}-$\A$-$\G$ (abbreviated as \CSM-$\A$-$\G$).

\def\conseq{\textcolor{gray}{\blacklozenge}}
\def\boe{\blacklozenge}
\def\css{**}
\newcommand{\kn}[1]{\textcolor{gray!55!black}{#1}}
\newcommand{\CSMDelAgMA}{\kn{\Pol$^{\conseq}$}}
\def\CSMDelAgMAref{\kn{[T\ref{thm:CSR-DelAg-MP+MA}]}}
\newcommand{\CSMDelAccSM}{\kn{\Pol$^{\conseq}$}}
\def\CSMDelAccSMref{\kn{[P\ref{prop:CSR-DelAcc-SM}]}}
\newcommand{\CSRAddAgExistsSM}{\newNPc}
\def\CSRAddAgExistsSMref{[T\ref{thm:CSR-AddAg-ExistsSM}]}
\newcommand{\CSMAddAgMA}{{\newNPc}}
\def\CSMAddAgMAref{[T\ref{thm:CSM-AddAg-MA}]}
\newcommand{\CSMAddAgPSM}{{\newNPc}}
\def\CSMAddAgPSMref{[T\ref{thm:CSM-AddAg-MA}]}
\newcommand{\CSRAddAgMA}{{\newNPc}}
\def\CSRAddAgMAref{[C\ref{cor:CSR-AddAg-MA}]}
\newcommand{\CSRDelAgMP}{\newPol}
\def\CSRDelAgMPref{[T\ref{thm:CSR-DelAg-MP+MA}]}
\newcommand{\CSRDelAgMA}{\newPol}
\def\CSRDelAgMAref{[T\ref{thm:CSR-DelAg-MP+MA}]}
\newcommand{\CSRAddAgSM}{\newNPc}
\def\CSRAddAgSMref{[T\ref{thm:CSR-AddAg-SM}]}
\def\CSRDelAccUSMref{[T\ref{thm:CSR-DelAcc-USM}]}

\def\CSMReorderMAref{[T\ref{thm:CSM-CSR-Reorder-MA}]}
\def\CSRReorderMAref{[T\ref{thm:CSM-CSR-Reorder-MA}]}
\def\CSMSwapMAref{[P\ref{prop:CSM-CSR-Swap-MA}]}
\def\CSRSwapMAref{[P\ref{prop:CSM-CSR-Swap-MA}]}
\def\CSRSwapSMref{[T\ref{thm:CSR-Swap-MS}]}
\def\CSRReorderSMref{[P\ref{prop:CSR-Reorder-MS}]}
\def\CSRReorderESMref{[T\ref{thm:CSR-bribery-ESM}]}
\def\CSRSwapESMref{[T\ref{thm:CSR-bribery-ESM}]}

\newcommand{\SMDelAgMA}{%
\Pol & [P.~\ref{prop:CSM-DelAg-MA},${\conseq}$]}
\newcommand{\SMDelAccSM}{%
\Pol & [P.~\ref{prop:CSR-DelAcc-SM}, ${\conseq}$]}
\newcommand{\SRAddAgExistsSM}{\NPc & [T.~\ref{thm:CSR-AddAg-ExistsSM}]}
\newcommand{\SMAddAgMA}{\NPc  &[T.~\ref{thm:CSM-AddAg-MA}]}
\newcommand{\SMAddAgPSM}{\NPc & [T.~\ref{thm:CSM-AddAg-MA}]}
\newcommand{\SRAddAgMA}{\NPc & [C.~\ref{cor:CSR-AddAg-MA}]}
\newcommand{\SRDelAgMP}{\Pol  &[T.~\ref{thm:CSR-DelAg-MP+MA}]}
\newcommand{\SRDelAgMA}{\Pol & [T.~\ref{thm:CSR-DelAg-MP+MA}]}
\newcommand{\SRAddAgSM}{\NPc  &[T.~\ref{thm:CSR-AddAg-SM}]}
\newcommand{\boecite}{[$\blacklozenge$]}
\newcommand{\MScite}{[$\heartsuit$]}
\newcommand{\trivialwcite}{\multicolumn{2}{c}{\trivial}}
\newcommand{\tancite}{[$\dag$]}
\newcommand{\BMMcite}{[$\spadesuit$]}
\newcommand{\ABMcite}{[$\star$]}
\newcommand{\new}[1]{\ensuremath{\boldsymbol{#1}}}

\newcolumntype{a}{>{}c} 
\def\trivial{---}
\def\NPc{\npc}
\def\newNPc{\textbf{NPc}}
\def\Pol{in $\p$}
\def\newPol{\textbf{in} $\boldsymbol{\p}$}
\def\own{\mathbf{A}}
\begin{table*}[t!]
\caption{Complexity results for \textsc{Control for Stable Marriage (\textup{or} Roommates)}-$\A$-$\G$ problems.
The problems \textsc{CSM}-$\A$-$\ESM$, marked ``---'', are trivial, as a stable matching always exists in an instance of SM.
``{\NPc}'' stands for $\np$-complete.
Hardness results for \textsc{CSR} that carry the mark of a result shown for \textsc{CSM} follow from that result, since \textsc{Stable Marriage} is a special case of \textsc{Stable Roommates} (membership in $\np$ is straightforward in each such case).
Additionally, we use the following notation:
\boecite\ and [$\conseq$] mark results by Boehmer et al.~\protect\cite{boe-bre-hee-nie:j:control-bribery-in-SM} and easy consequences thereof,  respectively;
\tancite\ marks results by Tan~\protect\cite{tan:j:maximum-stable-matching,tan:j:stable-partitions};
\MScite\ marks results by Mnich and Schlotter~\protect\cite{mni-sch:j:SM-covering-constraints};
\ABMcite\ marks results by Abraham et al.~\protect\cite{abr-bir-man:c:almost-stable-roommates}; 
\BMMcite\ marks results by  Bir{\'o} et al.~\protect\cite{bir-man-mit:j:size-vs-stability};
[$\css$] marks results by Chen, Skowron, and Sorge~\protect\cite{che-sko-sor:j:tradeoff}.
Previously known results are typeset in \kn{gray}; results new to this paper are typeset in \textbf{bold} black.
}
\label{tab:summary-for-SM}
\centering
 \begin{tabular}{@{}r@{\hspace*{2.2mm}}%
 w{c}{0.9cm}@{\hspace*{2.0mm}}%
 w{c}{0.9cm}@{\hspace*{2.0mm}}%
 w{c}{0.9cm}@{\hspace*{2.0mm}}%
 w{c}{0.9cm}@{\hspace*{2.0mm}}%
 w{c}{0.9cm}@{\hspace*{2.0mm}}%
 w{c}{0.9cm}@{\hspace*{2.3mm}}%
 w{c}{0.9cm}@{\hspace*{2.0mm}}%
 w{c}{0.9cm}@{\hspace*{2.0mm}}%
 w{c}{0.9cm}@{\hspace*{2.0mm}}%
 w{c}{0.9cm}@{\hspace*{2.3mm}}%
 c@{}}
   \toprule
   & \multicolumn{2}{c}{$\AddAg$} 
   & \multicolumn{2}{c}{$\DelAg$} 
   & \multicolumn{2}{c}{$\DelAcc$}
   & \multicolumn{2}{c}{$\Reorder$} 
   & \multicolumn{2}{c}{$\Swap$}\\
   & CSM & CSR  & CSM & CSR   & CSM & CSR & CSM & CSR   & CSM & CSR \\\midrule
   $\MA$ & \CSMAddAgMA & \CSRAddAgMA &  \CSMDelAgMA\ & \CSRDelAgMA & \kn{\NPc$^{\heartsuit}$} & \kn{\NPc$^{\heartsuit}$} & \newNPc & \newNPc & \newNPc & \newNPc \\
   & \CSMAddAgMAref & \CSRAddAgMAref &  \CSMDelAgMAref & \CSRDelAgMAref & & & \CSMReorderMAref & \CSRReorderMAref & \CSMSwapMAref & \CSRSwapMAref\\[5pt]
   $\MP$ & \kn{\NPc$^{\boe}$} & \kn{\NPc$^{\boe}$} &  \kn{\Pol$^{\boe}$} & \CSRDelAgMP & \kn{\NPc$^{\boe}$} & \kn{\NPc$^{\boe}$} & \kn{\NPc$^{\boe}$}  & \kn{\NPc$^{\boe}$} & \kn{\NPc$^{\boe}$} & \kn{\NPc$^{\boe}$} \\
    &  &  &  & \CSRDelAgMPref & &  \\[5pt]
   $\SM$ &  \kn{\Pol$^{\boe}$} & \CSRAddAgSM & \kn{\NPc$^{\boe}$} & \kn{\NPc$^{\boe}$}  &  \kn{\Pol$^{\boe}$} & \CSMDelAccSM & \kn{\Pol$^{\boe}$} & \newNPc & \kn{\Pol$^{\boe}$} & \newNPc \\
   &  & \CSRAddAgSMref &  &  &  & \CSMDelAccSMref & & \CSRReorderSMref & & \CSRSwapSMref\\[5pt]
   $\USM$ & \kn{\NPc$^{\boe}$} & \kn{\NPc$^{\boe}$} & \kn{\NPc$^{\boe}$} & \kn{\NPc$^{\boe}$} & \kn{\Pol$^{\boe}$} & $\newNPc$ & \kn{\NPc$^{\boe}$} & \kn{\NPc$^{\boe}$} & \kn{\NPc$^{\boe}$} & \kn{\NPc$^{\boe}$} \\
   & & & & & & \CSRDelAccUSMref 
   \\[5pt]
   $\ESM$ & \kn{\trivial} & \CSRAddAgExistsSM & \kn{\trivial} & \kn{\Pol$^{\dag}$} & \kn{\trivial} & \kn{\NPc$^{\star}$} & \kn{\trivial} & \newNPc & \kn{\trivial} & \newNPc \\
    & & \CSRAddAgExistsSMref & &  & & & & \CSRReorderESMref & & \CSRSwapESMref \\[5pt]
   $\PSM$ & \CSMAddAgPSM & \CSRAddAgExistsSM &  \kn{\Pol$^{\dag}$} & \kn{\Pol$^{\dag}$} & \kn{\NPc$^{\spadesuit}$} & \kn{\NPc$^{\spadesuit}$} & \kn{\NPc$^{\boe}$}& \kn{\NPc$^{\boe}$} & \kn{\NPc$^{\css}$} & \kn{\NPc$^{\css}$} \\
   & \CSMAddAgPSMref & \CSRAddAgExistsSMref &  &  & & \\
  \bottomrule
\end{tabular}
\end{table*}

\mypara{Our contributions.}
\Cref{tab:summary-for-SM} summarizes the computational complexity of all non-trivial action--goal--setting combinations.
The majority of these results are not new: they either appear explicitly in the prior literature on stable-matching control/manipulation, or they follow via transfer arguments (in particular, from \CSM\ to \CSR\ or by standard reductions from known ``almost stable'' and deletion problems).
Our main technical contributions are
\begin{inparaenum}[(i)]
  \item introducing and studying the control goal $\MA$ (ensuring that a designated agent is matched in some stable outcome), which was not part of the previously studied stable-matching control goal sets, and
  \item settling the remaining cases not implied by existing results or immediate transfers (highlighted in \Cref{tab:summary-for-SM}), yielding a complete and coherent landscape for these actions and goals in both \CSM\ and \CSR.
  For the tractable new cases we provide polynomial-time algorithms, and for the intractable ones we give matching $\np$-hardness proofs.
\end{inparaenum}

\myparaI{Outline.}
\Cref{sec:prelim} introduces notation and formalizes our control problems.
\Cref{sec:control,sec:bribery} contain our results for control actions ($\AddAg$, $\DelAg$, and $\DelAcc$), and for bribery actions ($\Reorder$ and $\Swap$), respectively.
We conclude in \cref{sec:conclusion}.
Proofs of results marked by ($\star$) are deferred to \ifshort the full version~\cite{che-sch:t:control-SM-arxiv}\else the appendices\fi.

\mypara{Related work.}
Control was introduced in computational social choice for elections by  Bartholdi III et al.~\cite{bar-tov-tri:j:control} and brought to stable matchings by Boehmer et al.~\cite{boe-bre-hee-nie:j:control-bribery-in-SM}, who study multiple actions/goals (including unique stability) for \textsc{Stable Marriage}.
A broader unifying perspective is provided by the IJCAI~2025 survey of Chen et al.~\cite{ChenKNRSS-ijcai25}, which reviews control across COMSOC (including voting, fair allocation, cooperative games, and matching under preferences) and helps situate stable-matching control among related structural-intervention models.

Several earlier lines of work coincide with special cases of our action/goal combinations under different motivations.
Tan~\cite{tan:j:maximum-stable-matching,tan:j:stable-partitions} studies maximum stable matchings for \SR, which corresponds to deleting as few agents as possible so that the remaining instance admits a perfect and stable matching (i.e., \CSR-$\DelAg$-$\PSM$).
Problems of the form \CSR-$\DelAcc$-$\mathcal{G}$ are closely related to \emph{almost stable} matchings:
Deleting $k$ acceptable pairs can be viewed as removing up to $k$ potential blocking pairs.
Abraham et al.~\cite{abr-bir-man:c:almost-stable-roommates} prove strong inapproximability for minimizing blocking pairs in \SR,
and Bir{\'o} et al.~\cite{bir-man-mit:j:size-vs-stability} show even stronger inapproximability for maximizing size while minimizing blocking pairs; their hardness already applies to the marriage setting with a perfect target, implying intractability for \CSM-$\DelAcc$-$\PSM$.
Mnich and Schlotter~\cite{mni-sch:j:SM-covering-constraints} study covering distinguished agents subject to a bound on blocking pairs; in particular, their hardness for covering even a single distinguished agent implies hardness for \CSM-$\DelAcc$-$\MA$.
Further intractability and parameterized results include Chen et al.~\cite{che-herm-sor-yed:c:-par-stable} (parameterized analysis for $\DelAcc$ variants),
Gupta and Jain~\cite{gup-jai:c:manipulation-matching} (weighted/destructive variants),
and Kamiyama~\cite{kam:j:super-stable-agent-deletion} (agent deletion to achieve super-stability with ties).

Finally, control by adding/deleting agents has recently been extended %
to coalition formation: Chen et al.~\cite{chen-gut-mus-sim:aamas26:control-hedonic} initiate the study of control in hedonic games under several stability notions and goals.
This underlines that participation-based interventions form a common algorithmic theme across COMSOC models; our results give a complete picture for the classical stable matching setting.

\section{Preliminaries}
\label{sec:prelim}

For a positive integer~$t$, we write $[t]\coloneqq\{1,\dots,t\}$.
In the following, we first introduce the basic concepts of the standard model for \textsc{Stable Roommates}, and then define the control problems we study.

\subsection{Model for \textsc{Stable Roommates}}
Formally, an instance~$I$ of \SR\ consists of a set~$U$ of agents and a collection $(\succ_u)_{u \in U}$ of preferences where each preference list~$\succ_u$ is a strict linear order over the set~$A(u) \subseteq U \setminus \{u\}$ of agents that $u$ finds \myemph{acceptable}.
We assume that the acceptability relation is symmetric, meaning that $x$ finds $y$ acceptable if and only if $y$ finds $x$ acceptable.
For the acceptability relation of each instance~$I$, there exists an \myemph{associated undirected graph~$G=(U,E)$} over the set of agents such that 
two agents are connected by an edge of~$E$ if and only if they find each other acceptable.
We use~$E(I)$ to denote the set of all acceptable pairs in~$I$.

A \myemph{matching} in~$I$ is a subset~$M \subseteq E$ of pairwise disjoint edges.
For a pair~$\{u,v\} \in M$, we will use the notation $M(u)=v$ and $M(v)=u$ to denote the partner of~$u$ and $v$, respectively. %
Consequently, we say that $M$ \myemph{covers}~$u$ and $v$. 
It is \myemph{perfect} for~$I$ if every agent in~$I$ is covered.
We write $M(u)=\emptyset$ to indicate that $u$ is \myemph{unmatched} by~$M$, that is, $M$ does not cover agent~$u$.
A \myemph{blocking pair} for some matching~$M$ in~$I$ is a pair~$\{u,v\} \in E$ of agents such that
\begin{compactenum}[(i)] 
  \item either $M(u)=\emptyset$ or $v \succ_u M(u)$, and similarly,
  \item either $M(v)=\emptyset$ or $u \succ_v M(v)$.
\end{compactenum}
A matching is \myemph{stable} if it admits no blocking pair.

Given an instance~$I=(U,(\succ_u)_{u\in U})$ of \SR\ and a subset~$W$ of agents, \myemph{$I-W$} is the instance obtained by \myemph{deleting the agents} in~$W$ from~$I$, that is, we let 
$I-W \coloneqq (U \setminus W,(\succ'_u)_{u\in U \setminus W})$ where $\succ'_u$ is the restriction of~$\succ_u$ to~$U \setminus W$.
Similarly, given a subset~$F \subseteq E$ of agent pairs,
the instance obtained from~$I$ by \myemph{deleting the acceptability} for the pairs in~$F$ is~\myemph{$I-F$} $\coloneqq (U,(\succ'_u)_{u\in U})$, where for each agent~$u \in U$, the set of acceptable agents becomes $A'(u)\coloneqq A(u) \setminus \{v\colon \{u,v\} \in F\}$
and thus $\succ'_u$ is the restriction of~$\succ_u$ to~$A'(u)$.

\mypara{\textsc{Stable Marriage}.}
The \textsc{Stable Marriage} model is the bipartite special case of \SR: The agent set is partitioned into two sides, and acceptability edges only go across the partition.
Following tradition, we refer to agents on the two sides as \emph{men} and  \emph{women}.
Each agent ranks its acceptable partners on the opposite side in strict order, and a matching pairs men and women.
Stability is defined identically (no blocking pair of mutually preferred partners), but unlike in \SR, a stable matching always exists in \textsc{Stable Marriage} and can be found by the Gale--Shapley algorithm~\cite{gal-sha:j:college-admissions-and-the-stability-of-marriage}.

\subsection{Problem definitions}\label{sec:problem-definitions}
\appendixsection{sec:problem-definitions}
We now formally define problem \CSR-$\A$-$\G$ for each action~$\A \in \{\AddAg$, $\DelAg$, $\DelAcc, \Swap, \Reorder\}$ and each goal~$\G \in \{\MA, \MP, \SM, \USM, \ESM,\PSM\}$.
\looseness=-1

\mypara{Adding agents.} We first focus on the control action of agent addition~$\AddAg$.
For this control action, the input not only contains the \myemph{original} set~$U$ of agents but also a set~$U'$ of \myemph{addable} agents, and the preferences of the agents are given for the whole market over $U \cup U'$.
For the control goals~$\ESM$ and~$\PSM$, the definitions are as follows:

\decprob{\CSR-$\AddAg$-$\ESM$ (\normalfont{resp.} \CSR-$\AddAg$-$\PSM$)}{
  An instance~$I=(U\cup U',(\succ_u)_{u\in U\cup U'})$ of \SR, where $U$
  denotes the set of \myemph{original} agents and $U'$ the set of \myemph{addable} agents, and an integer \myemph{budget}~$\ell$.}
{
  Is there a subset $W\subseteq U'$ with $|W|\le \ell$ such that $I-(U'\setminus W)$ admits a stable (resp., \emph{perfect} and stable) matching?
}

For the control goals $\MA$ and $\MP$, we require an additional input:

\decprob{\CSR-$\AddAg$-$\MA$ (\normalfont{resp.} \CSR-$\AddAg$-$\MP$)}{
  An instance~$I=(U\cup U',(\succ_u)_{u\in U\cup U'})$ of \textsc{Stable Roommates}, %
  an integer \myemph{budget}~$\ell$, and 
  an agent~${u^* \in U}$ (resp., a pair $e^* \in E(I)$ of agents).}
{
  Is there a subset $W\subseteq U'$ with $|W|\le \ell$ such that $I-(U'\setminus W)$ admits a stable matching that covers $u^*$ (resp., contains $e^*$)?
}

Finally, for the goals~$\SM$ and~$\USM$, following Boehmer et al.~\cite{boe-bre-hee-nie:j:control-bribery-in-SM}, the additional input is a \myemph{perfect} matching, which covers not only the original agents but also the addable agents.
The reason for this is that the control action~$\AddAg$ may alter the set of matched agents in an instance.

\decprob{\CSR-$\AddAg$-$\SM$ (\normalfont{resp.} \CSR-$\AddAg$-$\USM$)}{
  An instance~$I=(U\cup U',(\succ_u)_{u\in U\cup U'})$ of \textsc{Stable Roommates}, %
  an integer \myemph{budget}~$\ell$, and
  a perfect matching~$M$ of $I$.}
{
  Is there a subset $W\subseteq U'$ with $|W|\le \ell$ such that $I-(U'\setminus W)$ admits a stable matching $M'$ with $M'\subseteq M$ (resp., a \emph{unique} stable matching $M'$ with $M'\subseteq M$)?
}

\mypara{Deleting agents.}
The control problems regarding the control action of deleting agents are defined analogously, so we omit their definitions here, except for the goal~$\SM$, where we again assume the input matching to be perfect: 

\decprob{\CSR-$\DelAg$-$\SM$}{
  An instance~$I=(U,(\succ_u)_{u\in U})$ of \SR,
  an integer \myemph{budget}~$\ell$,
  and a perfect matching~$M$ of~$I$.}
{
Is there a subset~$W \subseteq U$ of at most~$\ell$ agents such that $I-W$ admits a stable matching~$M'$ with $M'\subseteq M$?
}

\mypara{Deleting acceptability.} The control problems \CSR-$\DelAcc$-$\G$ are defined analogously: instead of a set of at most~$\ell$ agents $W \subseteq U$, we ask for a set of at most~$\ell$ edges~$F \subseteq E(I)$ so that the instance~$I-F$ satisfies our goals.
For brevity, we omit these problem definitions.
We remark that, contrasting the previous two control actions, the input matching for \CSR-$\DelAcc$-$\SM$ is simply a matching, because deleting acceptability does not change the set of agents:

\decprob{\CSR-$\DelAcc$-$\SM$}{
  An instance~$I=(U,(\succ_u)_{u\in U})$ of \SR,
  an integer \myemph{budget}~$\ell$,
  and a matching~$M$ in~$I$.}
{
Is there a subset $F\subseteq E(I)$ with $|F|\le \ell$ such that $I-F$ admits a stable matching $M'$ with $M'\subseteq M$?
}

\mypara{Swapping agents in a preference list.}
We further consider the manipulative action~$\Swap$, introduced by Boehmer et al.~\cite{boe-bre-hee-nie:j:control-bribery-in-SM} as a \myemph{bribery} action: in contrast to the three previous actions, $\Swap$ does not change the set of agents nor the acceptability relation, but instead modifies the preference orders themselves.
Formally, given an instance~$I=(U,(\succ_u)_{u\in U})$ of \SR, a \myemph{swap operation} is a triple~$\sigma=(u,v,v')$ where $u \in U$ and $v,v' \in A(u)$ appear consecutively in~$\succ_u$, that is, $v\succ_u v'$ and no~$x \in A(u)$ satisfies~$v\succ_u x \succ_u v'$.
Applying~$\sigma$ to~$I$ yields the instance~$I_{\sigma}=(U,(\succ'_u)_{u\in U})$ in which~$\succ'_u$ is identical to~$\succ_u$ except that the relative order of~$v$ and~$v'$ is reversed, while $\succ'_{x}=\succ_x$ holds for each $x \in U \setminus \{u\}$.
Modeling a briber whose cost grows with the distance an agent must move in another agent's preferences, $\Swap$ is the most fine-grained of the bribery actions.
Given a sequence~$S$ of swap operations, we write \myemph{$I\oplus S$} for the instance obtained by applying the operations of~$S$ to~$I$ one after the other (each subsequent operation is interpreted in the already partially modified instance).

\mypara{Reordering preference lists.}
The second bribery action we adopt from Boehmer et al.~\cite{boe-bre-hee-nie:j:control-bribery-in-SM} is~$\Reorder$, which models a briber capable of fully ``buying'' an agent: a single operation may replace an agent's preference list with an arbitrary strict order over the same acceptable set.
Formally, a \myemph{reorder operation} is a pair~$\rho=(u,\succ^{\circ}_u)$ where $u \in U$ and~$\succ^{\circ}_u$ is a strict linear order over~$A(u)$.
Applying~$\rho$ to~$I$ replaces~$\succ_u$ with~$\succ^{\circ}_u$ and leaves every other agent's preference list intact.
Given a set~$R$ of reorder operations, with at most one per agent, we again write \myemph{$I\oplus R$} for the resulting instance.

\medskip
\noindent\emph{Restriction for the goals $\MA$ and $\MP$.}
As already observed by Boehmer et al.~\cite[Section~2.4]{boe-bre-hee-nie:j:control-bribery-in-SM} for the goal corresponding to~$\MP$, allowing the briber to act on the distinguished agents themselves trivializes the problem (e.g., a constant number of reorders on~$u^*$ and~$v^*$ would already make them each other's top choice).
They therefore explicitly forbid reordering the preference list of any agent in the distinguished pair, and we adopt the same restriction for the goal~$\MP$.
For the goal~$\MA$, however, protecting only the distinguished agent~$u^*$ itself is not enough: a single reorder operation placing $u^*$ at the top of the preference list of an arbitrary agent~$w \in A(u^*)$ already forces $u^*$ to be matched in every stable matching, as otherwise $\{u^*,w\}$ would form a blocking pair.
To prevent the briber from promoting the distinguished agent indirectly in this way, we additionally forbid operations on the agents that $u^*$ finds acceptable.
Formally, given an instance~$I$, a swap or reorder operation acting on an agent~$u$ is \myemph{admissible} if
\begin{compactitem}
\item for the goal~$\MA$ with distinguished agent~$u^* \in U$:\ ${u \neq u^*}$ and ${u \notin A(u^*)}$;
\item for the goal~$\MP$ with distinguished pair~${e^*=\{u^*,v^*\}}$:\ $u \notin \{u^*,v^*\}$.
\end{compactitem}
For the goals~$\SM$, $\ESM$, $\PSM$, and~$\USM$, every swap and reorder operation is admissible.

For each bribery action~$\A \in \{\Swap,\Reorder\}$ and each goal~$\G \in \{\MA,\MP,\SM$, $\ESM,\PSM,\USM\}$, the problem \CSR-$\A$-$\G$ asks whether the goal can be achieved by at most~$\ell$ admissible operations of type~$\A$.
We illustrate the formal definitions with two representative cases.

\decprob{\CSR-$\Swap$-$\MA$ (\normalfont{resp.} \CSR-$\Swap$-$\MP$)}{
  An instance~$I=(U,(\succ_u)_{u\in U})$ of \SR,
  an integer \myemph{budget}~$\ell$, and
  an agent~${u^* \in U}$ (resp., a pair $e^* \in E(I)$).}
{
Is there a sequence~$S$ of at most~$\ell$ admissible swap operations on~$I$ such that $I\oplus S$ admits a stable matching that covers~$u^*$ (resp., contains~$e^*$)?
}

\decprob{\CSR-$\Reorder$-$\ESM$ (\normalfont{resp.} \CSR-$\Reorder$-$\PSM$)}{
  An instance~$I=(U,(\succ_u)_{u\in U})$ of \SR\
  and an integer \myemph{budget}~$\ell$.}
{
Is there a set~$R$ of at most~$\ell$ reorder operations on~$I$ such that $I\oplus R$ admits a stable (resp., \emph{perfect} and stable) matching?
}

We omit the analogous definitions for the other action--goal combinations.
The marriage variants of our problems are defined analogously, with the input containing an instance of \textsc{Stable Marriage} instead of \SR.

We conclude this section with the following simple observation.
\begin{restatable}[\linkproof{obs:MA->MP}]{observation}{obsMAMP}
  \label{obs:MA->MP}
  For every control action $\mathcal{A} \in \{\AddAg$, $\DelAg$, $\DelAcc\}$,
  there is a polynomial-time Turing reduction from
  \CSR-$\mathcal{A}$-$\MA$ to \CSR-$\mathcal{A}$-$\MP$, and from
  \CSM-$\mathcal{A}$-$\MA$ to \CSM-$\mathcal{A}$-$\MP$.
\end{restatable}

We remark that \cref{obs:MA->MP} does not extend directly to the bribery actions $\Swap$ and~$\Reorder$, since the admissibility restrictions of the two problems differ: an operation that is admissible for the goal~$\MP$ with distinguished pair~$\{u^*,w^*\}$ may act on an agent in~$A(u^*)\setminus\{w^*\}$ and is then inadmissible for the goal~$\MA$ with distinguished agent~$u^*$.

\prooftoappendix{obs:MA->MP}{\obsMAMP*}{
  We give the reduction for \CSR; the \CSM\ case is identical.
  Note that since $\mathcal{A}$ is a control action, no admissibility restrictions apply, and hence the \CSR-$\mathcal{A}$-$\MA$ instance and the \CSR-$\mathcal{A}$-$\MP$ instances constructed below have the same set of allowed actions.
  Let $I$ be an instance of \CSR-$\mathcal{A}$-$\MA$ with designated agent~$u^*$.
  For each $w^* \in A(u^*)$, let $I_{w^*}$ denote the instance of
  \CSR-$\mathcal{A}$-$\MP$ obtained from $I$ by replacing the designated
  agent~$u^*$ with the designated pair $e^*=\{u^*,w^*\}$.
  Then $I$ is a yes-instance of \CSR-$\mathcal{A}$-$\MA$ if and only if some
  $I_{w^*}$ is a yes-instance of \CSR-$\mathcal{A}$-$\MP$:
  the forward direction follows by taking $w^*$ to be $u^*$'s partner in a
  witnessing stable matching, and the reverse direction is immediate.
}

\section{Control actions}
\label{sec:control}
\appendixsection{sec:control}

We first focus on control by adding or deleting agents, or deleting acceptability.
Within the section, we organize our results based on the goal we aim to achieve.

\mypara{Matching an agent or a pair of agents.}
We start with our main algorithmic result which shows that even in an instance of \textsc{Stable Roommates}, we can efficiently decide whether, by deleting at most~$\ell$ agents, we can obtain an instance in which some stable matching pairs two designated agents together. Naturally, this also implies an algorithm for ensuring via agent deletions that a given agent is matched in some stable matching---it suffices to try all its possible partners.

Boehmer et al.~\cite{boe-bre-hee-nie:j:control-bribery-in-SM} showed that, in the marriage setting, it is polynomial-time solvable to make a given pair part of some stable matching.
From this, 
the polynomial-time solvability of \textsc{CSM}-$\DelAg$-$\MA$ follows using \Cref{obs:MA->MP}.
We now extend the approach by Boehmer et al.~\cite{boe-bre-hee-nie:j:control-bribery-in-SM} to the roommates setting.
The crucial idea is to use the concept of stable partitions by Tan~\cite{tan:j:stable-partitions} and a result by Tan and Hsu~\cite{tan-hsu:j:-generalization-of-the-stable-matching-problem} regarding deleting an agent; see also \cite[Sections~4.3.2 and~4.3.4]{man:b:algorithmics-of-matching-under-preferences} for the necessary background on stable partitions.

\begin{theorem}\label{thm:CSR-DelAg-MP+MA}
\textsc{CSR}-$\DelAg$-$\MP$ and \textsc{CSR}-$\DelAg$-$\MA$ are in~$\p$.
\end{theorem}

\def\J{\mathcal{J}}
\begin{proof}
We present a polynomial-time algorithm for the \textsc{CSR}-$\DelAg$-$\MP$ problem. 
By \Cref{obs:MA->MP}, 
this implies the polynomial-time solvability of \textsc{CSR}-$\DelAg$-$\MA$ as well.
As already mentioned, we will use the structural properties of stable partitions  and ideas from the algorithm given by Boehmer et al.~\cite{boe-bre-hee-nie:j:control-bribery-in-SM} for the bipartite version of our problem, \textsc{CSM}-$\DelAg$-$\MP$.

Let $\J=(I,\ell, \{a,b\})$ be our input instance where
$I$ is an instance of \SR, 
$\ell$ is our budget, i.e., the number of agents we are allowed to delete, and
$\{a,b\}$ is the (acceptable) pair in~$I$ that we aim to include in the desired matching. We say that a set~$S$ of agents is a \myemph{solution} for~$\J$ if removing $S$ from~$I$ yields an instance, denoted by $I-S$, where some stable matching contains~$\{a,b\}$.

Let $A^\star$ and $B^\star$ be the sets of agents which are preferred by~$a$ to~$b$, or by~$b$ to~$a$, respectively (note that $A^\star$ and~$B^\star$ need not be disjoint).
Furthermore, let the edge set~$F$ contain those pairs $\{x,y\}$, where 
\begin{compactenum}[(i)]
  \item $x\in A^{\star}$, and $y=a$ or $x$ prefers $a$ to~$y$, or 
  \item $x\in B^{\star}$, and $y=b$ or $x$ prefers $b$ to~$y$.
\end{compactenum}
Formally, let
\begin{equation}
\label{def:A*B*}
\begin{array}{rlrl}
  A^{\star} & = \{x\colon x \succ_a b\}, \qquad
  & 
   F & =  
\{\{x,y\} \colon x\in A^{\star}, (y=a \textup{ or } a \succ_x y)\} 
  \\
    B^{\star} & =  \{x\colon x \succ_b a\},    
    & 
  & \phantom{mm} \cup 
\{\{x,y\}\colon x\in B^{\star}, (y=b \textup{ or } b \succ_x y)\}.
\end{array}
\end{equation}
Note that no pair in~$F$ can be included in a stable matching containing~$\{a,b\}$, because all such pairs either match~$a$ or~$b$ (but not via $\{a,b\}$), or result in a blocking pair formed by one of the agents in the pair and an agent in~$\{a,b\}$.

Let $I^\star$ denote the instance obtained from~$I$ by deleting all pairs in~$F$ from the set of acceptable pairs,
i.e., $I^{\star}\coloneqq I-F$.

\begin{algorithm}[t!] 
  \caption{Algorithm for solving \CSR-$\DelAg$-$\MP$}\label{alg:CSR-DelAg-MP} 
  \KwI{An instance~$I$ of \SR, a budget~$\ell$, and a pair of agents~$\{a,b\}$}

  \KwO{``yes'' if there exists a set~$S$ of at most~$\ell$ agents such that $I-S$ admits a stable matching that contains~$\{a,b\}$; otherwise ``no''}

  Compute~$A^{\star}$, $B^{\star}$, and $F$ as defined in \eqref{def:A*B*} 

  Compute instance~$I^{\star}\coloneqq I-F$

  Compute a stable partition~$\Pi$ for~$I^\star$ \label{alg:stablepartition}

  Let $U_0$ denote the set of singleton parties in~$\Pi$ consisting of an agent in~$A^\star \cup B^\star$

  Let $r$ denote the number of odd parties of size at least~$3$ in~$\Pi$

  Return ``yes'' if $r+|U_0| \leq \ell$; otherwise, return ``no.'' \label{alg:CSR-DelAg-MP-return}
\end{algorithm}

Our algorithm is depicted in \Cref{alg:CSR-DelAg-MP}. %
Note that line~\ref{alg:stablepartition} can be performed in linear time using Tan's algorithm~\cite{tan:j:stable-roommates} for computing a stable partition. 
Hence, the algorithm runs in polynomial time. It remains to show its correctness. 
We start with the following claim.

\begin{clm}
\label{clm:characterize-sols}
 A set $S$ of agents with $S \cap \{a,b\}=\emptyset$ is a solution for~$\J$ if and only if $I^\star-S$ admits a stable matching that matches every agent in~$(A^\star \cup B^\star) \setminus S$.
\end{clm}

\begin{claimproof}
Assume that $S$ is a solution for~$\J$. Then there exists a stable matching~$M$ in~$I-S$ containing $\{a,b\}$. First, we show that $M \cap F=\emptyset$. Suppose for the sake of contradiction that $\{x,y\} \in M \cap F$. Then $x \in A^\star \cup B^\star$; by symmetry we may assume $x \in A^\star$. By the definition of~$F$, we infer that either $y=a$ or $x$ prefers~$a$ to~$y$. The former is not possible, because $M(a)=b \notin A^\star$; the latter is not possible either, because it would mean that $\{x,a\}$ blocks~$M$: $x$ prefers~$a$ to~$y=M(x)$ and $a$ prefers~$x$ to~$M(a)=b$ by the definition of~$A^\star$. This contradiction proves that $M$ is a matching disjoint from~$F$, i.e., it is a matching in~$I^\star-S$. 

It remains to show that $M$ matches every agent in the set~$(A^\star \cup B^\star) \setminus S$. Suppose for the sake of contradiction that $x \in (A^\star \cup B^\star) \setminus S$ is unmatched by~$M$. By symmetry, we may assume $x \in A^\star$. However, then the pair $\{x,a\}$ is a blocking pair, a contradiction. Thus, $M$ has the desired property.

For the other direction, assume that for some agent set~$S$ with ${S \cap \{a,b\}=\emptyset}$, the instance $I^\star -S$ admits a stable matching~$M$ that matches every agent in the set~$(A^\star \cup B^\star) \setminus S$. 
In this case, $M$ matches $a$ to~$b$ because they are each other's best choice in~$I^\star$ due to the deletion of the acceptability of all pairs in~$F$. Moreover, $M$ is stable in~$I-S$, because it is stable in~$I^\star-S$ and, since all agents in $(A^\star \cup B^\star) \setminus S$ are matched in~$I^\star-S$, no pair in~$F$ blocks~$M$. 
\end{claimproof}
 
To show the correctness of our algorithm, assume first that it returns ``yes''. Let $S$ contain an arbitrarily chosen agent from each odd party of size at least~$3$ in~$\Pi$ and all singleton parties consisting of an agent in~$A^\star \cup B^\star$.
Note that $a,b \notin S$: since $a$ and~$b$ are each other's top choice in~$I^\star$, they form a party of size two in~$\Pi$, so neither belongs to an odd party.
Then deleting $S$ from~$I^\star$ yields an instance~$I'$ whose stable partition~$\Pi'$ has no odd parties of size at least~$3$ and whose singleton parties are all disjoint from~$A^\star \cup B^\star$, due to well-known properties of stable partitions~\cite{tan:j:stable-partitions} (see also \cite[Proposition~4.11]{man:b:algorithmics-of-matching-under-preferences}). This means that in all stable matchings for~$I^\star-S$, the set of unmatched agents (i.e., those agents that form singletons in~$\Pi'$) is disjoint from~$A^\star \cup B^\star$. By \cref{clm:characterize-sols}, $S$ is therefore a solution, and since $|S|=r+|U_0| \leq \ell$ by line~\ref{alg:CSR-DelAg-MP-return} of \cref{alg:CSR-DelAg-MP}, we obtain that $\J$ is a ``yes''-instance.

It remains to prove that whenever $S$ is a solution of size at most~$\ell$, then the algorithm returns ``yes''.
Clearly, $a,b \notin S$, as some stable matching of~$I-S$ must contain the pair~$\{a,b\}$.
Since $S$ is a solution, by \cref{clm:characterize-sols} we know that $I^\star-S$ admits a stable matching that matches all agents in~$(A^\star \cup B^\star) \setminus S$. This means that the stable partition~$\Pi_S$ of~$I^\star-S$ contains no odd parties of size at least~$3$, and contains no singleton odd parties formed by some agent in~$(A^\star \cup B^\star) \setminus S$; let us call such odd parties \myemph{forbidden}. However, it is known that deleting an agent from an instance of \SR\ results in an instance whose stable partitions contain all odd parties of the original instance except, possibly, for one (see \cite[Corollary~3.9]{tan-hsu:j:-generalization-of-the-stable-matching-problem}).
Since $\Pi$ has exactly $|U_0|+r$ forbidden odd parties, deleting fewer than~$r+|U_0|$ agents from~$I^\star$ will result in an instance that has at least one forbidden odd party. This implies $\ell \geq |S| \geq r+|U_0|$, so the algorithm returns ``yes'' on line~\ref{alg:CSR-DelAg-MP-return}.
\end{proof}

%

Let us turn our focus to control by adding agents.
Our next result shows that making an agent matched~($\MA$) in a stable matching is no easier than making a pair matched~($\MP$).
The reduction is from the $\np$-complete \textsc{Clique} problem~\cite{gar-joh:b:int} and is inspired by the one given by Boehmer et al.~\cite{boe-bre-hee-nie:j:control-bribery-in-SM} for the $\MP$ case, and happens to prove hardness also for \textsc{CSM}-$\AddAg$-$\PSM$.
 
\begin{restatable}[\linkproof{thm:CSM-AddAg-MA}]{theorem}{thmCSMAddAgMA}    
\label{thm:CSM-AddAg-MA}
\textsc{CSM}-$\AddAg$-$\MA$ 
and \textsc{CSM}-$\AddAg$-$\PSM$ are $\np$-complete.
\end{restatable}
\prooftoappendix{thm:CSM-AddAg-MA}{\thmCSMAddAgMA*}{
  It is clear that both problems are in~$\np$:
  Given a set of agent additions resulting in a modified instance~$I'$, it suffices to compute a single stable matching in~$I'$ and check whether it is perfect or whether it matches our given agent,
  because by the well-known Rural Hospitals Theorem~\cite{irv-lea-gus:j:efficient-algorithm-for-optimal-stable-marriage},
  every stable matching in~$I'$ covers the same set of agents.

  As mentioned, we present a polynomial-time many-to-one reduction that can be obtained as a modification of the reduction by Boehmer et al.~\cite[Theorem~1]{boe-bre-hee-nie:j:control-bribery-in-SM}; the modified reduction shows $\np$-hardness for both \textsc{CSM}-$\AddAg$-$\MA$
  and \textsc{CSM}-$\AddAg$-$\PSM$. 
  For convenience, we present the modified reduction in full detail. 

  We reduce from the $\np$-complete \textsc{Clique} problem~\cite{gar-joh:b:int}.

 \decprob{Clique}
  {An undirected graph~$G$ and a non-negative integer~$k$.}
  {Does~$G$ admit a \myemph{clique} of size at least~$k$, i.e., a complete subgraph~$K$ with at least~$k$ vertices?}

  Let $G=(V,E)$ be the input graph and~$k$ the size of the desired clique.
  We construct an instance~$I$ of \textsc{Stable Marriage} as follows. For each vertex~$v \in V$, we introduce a woman~$w_v$ and a man~$m'_v$.
  For each edge~$e \in E$, we introduce a woman~$w_e$ and two men~$m_e$ and $m'_e$.
  We further create a set $S=\{s_1,\dots,s_{\binom{k}{2}}\}$ of   \myemph{selector women}, a set $D=\{d_1,\dots,d_{|V|-k}\}$ of  \myemph{dummy men}, a woman~$w^*$, and a man~$m^*$. 
The preferences of the agents are as follows (with each set in the preference lists interpreted as an arbitrary but fixed, strict linear order of its elements): 
\[
\begin{array}{llllll}
w^*\colon & m^*, & & m^*\colon & S\succ w^*; \\
s\colon & \{m_e\colon e \in E\}\succ m^* & \forall s \in S;\qquad~  &
m_e\colon & w_e\succ \{w_v\colon v \in e\}\succ S & \forall e \in E; \\
w_e\colon & m'_e \succ m_e & \forall e \in E; &
m'_e\colon & w_e & \forall e \in E; \\
w_v\colon & m'_v\succ \{m_e\colon e \ni v\}\succ D & \forall v \in V;&
m'_v\colon & w_v &  \forall v \in V; \\
&&&d\colon & \{w_v\colon v \in V\} & \forall d \in D.
\end{array}
\]

We let~$M'_V\coloneqq \{m'_v\colon v \in V\}$ and $M'_E \coloneqq \{m'_e\colon e\in E\}$ and set $M'_V \cup M'_E$ as the set of addable agents, with the remaining agents forming the set of original agents. 
Note that the addable agents are all men.
Our goal in \textsc{CSM}-$\AddAg$-$\MA$ is to match agent~$w^*$, while our goal in \textsc{CSM}-$\AddAg$-$\PSM$ is to obtain a perfect and stable matching. 
Finally, we set our budget, that is, the number of agents from $M'_V \cup M'_E$ that we can add to~$I$, as $\ell \coloneqq k+\binom{k}{2}$. 

To see the correctness of the reduction, first assume that
$G$ admits a clique~$K$ of order~$k$; let $V(K)$ and $E(K)$ denote its vertices and edges, respectively. 
We show that adding the $k+\binom{k}{2}$ men $\{m'_v\colon v  \in V(K)\} \cup \{m'_e\colon e\in E(K)\}$ yields an instance~$I_K$ that admits a perfect and stable matching~$M$ that, in particular, matches~$w^*$ as well. Clearly, every stable matching in the resulting instance~$I_K$ contains all pairs in $F'\coloneqq\{\{m'_v,w_v\}\colon v  \in V(K)\} \cup \{\{m'_e,w_e\}\colon e\in E(K)\}$, 
because $F'$  matches every involved agent to its best choice. We construct $M$ as follows.
We match $w^*$ to $m^*$, match all selector women to the men in $\{m_e\colon e \in E(K)\}$ in a stable way, add the pairs in the matching~\myemph{$F$} $\coloneqq \{\{m_e,w_e\}\colon e\in E \setminus E(K)\}$ together with those in~$F'$ to~$M$, 
and finally, we match the women in~$\{w_v\colon v \in V \setminus V(K)\}$ to the $|V|-k$ dummy men in a stable way.\footnote{\label{fn:stableway}By matching a set $X$ of women with a set~$Y$ of men \myemph{in a stable way}, we mean matching them according to an arbitrary stable matching in the SM instance induced by $X \cup Y$. Note that we only do this when preferences in this sub-instance are complete, and hence, each of its stable matchings matches all agents in~$X \cup Y$.}

It is straightforward to verify that $M$ covers every agent and is stable: 
\begin{compactitem}
\item No blocking pair contains the distinguished woman~$w^*$, because she is matched with her best choice. %
\item No blocking pair contains a selector woman~$s \in S$ and some man $m_e$: If $e \notin E(K)$, then this follows from $\{m_e,w_e\} \in F$, otherwise we have $M(m_e) \in S$, and thus $\{s,m_e\}$ does not block~$M$ due to the stability of the restriction of $M$ to~$S \cup \{m_e\colon e\in E(K)\}$.
\item No blocking pair contains woman~$w_e$ corresponding to an edge~$e$, because she is matched to her best choice in~$I_K$: If $e \in E(K)$, then she is matched to~$m'_e$; otherwise, she is matched to~$m_e$ and $m'_e$ is not present.

\item No blocking pair contains a woman~$w_v$:
If $v \in V(K)$, then she is matched to~$m'_v$, her best choice.
If $v \notin V(K)$, then $M(w_v) \in D$.
In this case, she cannot block with any other dummy agent  due to the stability of the restriction of $M$ to~$D \cup \{w_v\colon v \in V \setminus V(K)\}$.
Moreover, each non-dummy man acceptable to~$w_v$ is a man~$m_e$ corresponding to an edge incident to~$v$ and hence not in $E(K)$, which means that $m_e$ is matched to his best choice by $\{m_e,w_e\} \in F$. 
\end{compactitem}
Since no woman is involved in a blocking pair, $M$ is indeed a stable matching in~$I_K$.

For the other direction, assume that by adding a set of at most~$k+\binom{k}{2}$ men from~$M'_V \cup M'_E$ to~$I$, we obtain an instance~$I'$ that admits a stable matching~$M$ in which $w^*$ is matched (which always holds if $I'$ admits a perfect and stable matching).
We first show that in this case we may assume w.l.o.g.\ that $M$ contains no pair $\{m_e,w_v\}$ for some edge~$e$ incident to some vertex~$v$ in~$G$; we call such pairs \myemph{problematic}. Indeed, if ${\{m_e,w_v\} \in M}$, then since $\{m_e,w_e\}$ cannot block~$M$, we know that ${M(w_e)=m'_e}$. Then we modify~$I'$ and~$M$  by removing the addable agent~$m'_e$ from~$I'$ and adding~$m'_v$ instead, setting also
$M'={\big(M \setminus \{\{m_e,w_v\}, \{m'_e, w_e\}\}\big) \cup \{\{m_e,w_e\},\{m'_v, w_v\}\}}$. It is clear that the resulting instance~$I''$ can still be obtained by adding at most $k+\binom{k}{2}$ agents to~$I$, and $M'$ is a stable matching in~$I''$ covering~$w^*$ and containing fewer problematic pairs than~$M$. Hence, repeating this operation as long as necessary, we end up with an instance~$\hat{I}$ and a stable matching~$\hat{M}$ in~$\hat{I}$ that covers~$w^*$ but contains no problematic pairs.
Let $X_V$ and~$X_E$ denote the set of vertices and edges in~$G$ that correspond to the addable men that are present in~$\hat{I}$.

Observe first that $w^*$ can only be matched to~$m^*$, that is, $\{m^*,w^*\} \in \hat{M}$.
This means that all selector women must be matched in~$\hat{M}$, as an unmatched selector woman would form a blocking pair with~$m^*$. Hence, there are exactly $\binom{k}{2}$ men in~$\{m_e\colon e\in E\}$ matched to selector women; let $P$ denote the set of these men.
Note that if $m_e \in P$, then $m'_e$ must be present in~$I'$ and matched to~$w_e$ by~$\hat{M}$ (that is, ${e \in X_E}$), as otherwise $\{m_e,w_e\}$ would block~$\hat{M}$.
Furthermore, for each $m_e\in P$, each woman in~$\{w_v\colon v  \in e\}$ must also be matched by~$\hat{M}$ to
a man whom she prefers to~$m_e$, as otherwise $\{m_e,w_v\}$ would block~$\hat{M}$.
Since~$\hat{M}$ contains no problematic pair, this means that for each $m_e\in P$ and each endpoint~$v \in e$,
the man~$m'_v$ must be present in~$I'$ and matched to~$w_v$ (that is, $v \in X_V$).
Hence, for each edge~$e$ in $\hat{E}\coloneqq \{e\colon m_e \in P\} \subseteq X_E$, both endpoints of~$e$ must be in~$X_V$. By $|X_E|+|X_V| \leq \binom{k}{2}+k$, we get that the $\binom{k}{2}$ edges in~$\hat{E}$ must have altogether at most $k$ endpoints, i.e., they must form a clique of size~$k$ in~$G$.
}

Note that \textsc{CSR}-$\AddAg$-$\MA$ is in~$\np$, because even in an instance of \SR, all stable matchings cover the same set of agents~\cite[Theorem 4.5.2]{gus-irv:b:the-stable-marriage-problem}.
This leads to the following consequence of \cref{thm:CSM-AddAg-MA}.
\begin{corollary}\label{cor:CSR-AddAg-MA}
\textsc{CSR}-$\AddAg$-$\MA$ 
is $\np$-complete.
\end{corollary}

\mypara{Stabilizing a matching.}
Having settled the complexity of using control actions to achieve the goals~$\MA$ and~$\MP$, we next turn to the goal~$\SM$ of stabilizing a prescribed matching.
For \textsc{Stable Marriage}, $\AddAg$-$\SM$ is solvable in polynomial time by Boehmer et al.~\cite{boe-bre-hee-nie:j:control-bribery-in-SM}.
The roommates setting is more brittle: even when we are given a perfect ``plan'' matching and are only allowed to keep pairs from that plan, deciding whether agent additions suffice becomes intractable.

\begin{restatable}[\linkproof{thm:CSR-AddAg-SM}]{theorem}{thmCSRAddAgSM}
\label{thm:CSR-AddAg-SM}
\textsc{CSR}-$\AddAg$-$\SM$ is $\np$-complete.
\end{restatable}
\prooftoappendix{thm:CSR-AddAg-SM}{\thmCSRAddAgSM*}{
  It is clear that the problem is in $\np$, so we give a polynomial-time many-to-one reduction from the following $\np$-complete problem~\textsc{Independent Set}~\cite{gar-joh:b:int} to prove its $\np$-hardness. 
  \decprob{Independent Set}
  {An undirected graph~$G$ and a non-negative integer~$k$.}
  {Does~$G$ admit a subset of at least~$k$ vertices that form an \myemph{independent set}, that is,
    these vertices are pairwise non-adjacent?}

Let $G=(V,E)$ and $k$ be our input of \textsc{Independent Set}. 

For each vertex~$v \in V$, we introduce three agents~$a_v,b_v,c_v$ and their respective copies $a'_v,b'_v,c'_v$.
We let the input matching~$M$ match each agent in~$A=\{a_v,b_v,c_v\colon v \in V\}$ with its copy in $A'=\{a'_v,b'_v,c'_v\colon v \in V\}$. 
Let the preferences of the agents be defined by setting the following for each $v \in V$:
\[
\begin{array}{llll}
a_v\colon & a'_v\succ b_v \succ c_v; \qquad\qquad
  & a'_v\colon & \{a'_u\colon\{u,v\} \in E\}\succ a_v; \\ %
b_v\colon & b'_v\succ a_v;
  & b'_v\colon & b_v; \\
c_v\colon & c'_v\succ a_v;
  & c'_v\colon & c_v. \\
\end{array}
\]
We set $A'$ as the set of addable agents, and we set $\ell=2|V|-k$ as our budget.

First assume that $G$ contains a set~$X$ of~$k$ vertices that form an independent set. Let $I_X$ denote the instance obtained by adding the agents in \[A'_X\coloneqq\{a'_v\colon v \in X\} \cup \{b'_v,c'_v\colon v \in V \setminus X\}.\]
Thus, $I_X$ is the instance induced by the agent set $A \cup A'_X$. 
Note that $|A'_X|=|X|+2|V \setminus X|=\ell$.

Define the matching \[M_X\coloneqq\{\{p,p'\}\colon p \in A,p' \in A'_X\}\] in~$I_X$.
Observe that $M_X \subseteq M$ is stable in~$I_X$:
\begin{itemize}
\item No agent~$a_v$ is contained in a blocking pair: If $v \in X$, then $a_v$ is matched to~$a'_v$, its best choice in~$I_X$.
If $v \in V \setminus X$, then $b_v$ and $c_v$ are the only agents acceptable for~$a_v$ in~$I_X$, and they are both matched to their respective copies, which are their respective best choices in~$I_X$.
\item Neither $b_v$ nor $c_v$ is contained in a blocking pair: If $v \in X$, then $a_v$ is matched to its best choice (its copy~$a'_v$), so does not form a blocking pair with~$b_v$ or~$c_v$; if $v \in V \setminus X$, then both $b_v$ and $c_v$ are matched to their copies, which are their respective best choices.
\item Agents~$b'_v$ and $c'_v$, $v\in V\setminus X$, cannot be involved in any blocking pair since each of them is matched with its best choice.
\item No blocking pair contains two agents in~$A'_X$, say $a'_v$ and $a'_u$: if $u,v \in X$, then $a'_v$ and $a'_u$ are not acceptable to each other, because $\{u,v\} \notin E$ since $X$ is an independent set.
\end{itemize}
This shows that our instance of \textsc{CSR}-$\AddAg$-$\SM$ is a ``yes''-instance.

For the other direction, assume that there is a set $A^\star \subseteq A'$ of at most $\ell$ agents such that the instance~$I'$ induced by the agents in~$A \cup A^\star$ contains a matching $M' \subseteq M$ that is stable in~$I'$.
We claim that $a'_v \in A^\star$ or $\{b'_v,c'_v\} \subseteq A^\star$ holds for each vertex $v \in V$. 
To see this, assume $a'_v \notin A^\star$. This means that $a_v$ is not matched by~$M'$. 
Hence, $M'$ must match both $b_v$ and~$c_v$, as otherwise $\{a_v,b_v\}$ or $\{a_v,c_v\}$ would block~$M'$. Therefore, we get $\{b'_v,c'_v\} \subseteq A^\star$, as required.

Let us define $X\coloneqq\{v\colon a'_v \in A^\star\}$. Note that \[2|V|-k=\ell \geq |A^\star| \geq |X|+2 |V \setminus X|=2|V| -|X|\]
by the previous paragraph, so we get $|X| \geq k$. It remains to see that $X$ is an independent set: indeed, if two vertices~$u$ and~$v$ of~$X$ are adjacent in~$G$, then $M'$ matches $a_u$ and~$a_v$ with their respective copies; however, then $\{a'_v,a'_u\}$ is a blocking pair for~$M'$.
Hence, $X$ is an independent set of size at least~$k$ in~$G$, proving the correctness of our reduction.
}

Boehmer et al.~\cite[Observation~1]{boe-bre-hee-nie:j:control-bribery-in-SM} observed that for the marriage case, the number of blocking pairs of a given matching~$M$ lower-bounds the number of acceptable pairs to delete to make~$M$ stable.
Since this observation applies to the roommates case as well,
the following holds. 

\begin{restatable}[\linkproof{prop:CSR-DelAcc-SM}]{proposition}{propCSRDelAccSM}
  \label{prop:CSR-DelAcc-SM}
  \textsc{CSR}-$\DelAcc$-$\SM$ can be solved in linear time.
\end{restatable}

\prooftoappendix{prop:CSR-DelAcc-SM}{\propCSRDelAccSM*}{%
  We state the proof for the sake of completeness.
  Let our input for \CSR-$\DelAcc$-$\SM$ be an instance~$I$ of \SR, a matching~$M$ in~$I$, and our budget~$\ell$ on the number of acceptability deletions we can perform. Note that $(I,\ell, M)$ is a ``yes''-instance of \CSR-$\DelAcc$-$\SM$ if and only if the number of blocking pairs for~$M$  is at most~$\ell$: Clearly, we need to delete all blocking pairs to make~$M$ stable, and deleting all of them is sufficient to ensure that $M$ becomes stable. Since we can compute the number of blocking pairs for~$M$ in linear  time, the result follows.
}

\mypara{Ensuring the existence of a (perfect) stable matching.}
We now move to the existence-oriented goals for $\AddAg$ in \CSR, where no reference matching is given and the controller only aims to ensure that a stable (or perfect and stable) matching exists after adding agents.
Even under this additional freedom, adding agents remains computationally hard in the roommates setting.

\begin{restatable}[\linkproof{thm:CSR-AddAg-ExistsSM}]{theorem}{thmCSRAddAgExistsSM}
\label{thm:CSR-AddAg-ExistsSM}
\textsc{CSR}-$\AddAg$-$\ESM$ and \textsc{CSR}-$\AddAg$-$\PSM$ are $\np$-complete.
\end{restatable}
\prooftoappendix{thm:CSR-AddAg-ExistsSM}{\thmCSRAddAgExistsSM*}{
It is clear that both problems are in $\np$, because checking whether an  instance~$I$ of SR admits a stable matching can be done in polynomial time~\cite{irv:j:efficient-algorithm-for-stable-roommates-problem}, and since all stable matchings in~$I$ have the same size~\cite{gus-irv:b:the-stable-marriage-problem}, finding any stable matching by, e.g., Irving's algorithm~\cite{irv:j:efficient-algorithm-for-stable-roommates-problem} suffices to check whether a perfect and stable matching exists.

We present a reduction from the $\np$-complete \textsc{Independent Set} problem~\cite{gar-joh:b:int} which shows the $\np$-hardness of both problems. Let the graph $G=(V,E)$ and the integer $k$  be our input for \textsc{Independent Set}. For each vertex $v \in V$, let $N(v)$ denote the set of vertices adjacent to~$v$ in~$G$.

We construct an instance~$I$ of \SR\ over agent set $A\coloneqq V \cup \{s_i,a_i,b_i\colon i\in[k]\}$. The preferences of these agents are as follows (with each set in the preference lists interpreted as an arbitrary but fixed, strict linear order of its elements):
\[
\begin{array}{lll}
v\colon & N(v)\succ  s_1\succ \dots \succ s_k
  & \textrm{ for each $v \in V$;} \\
s_i\colon & V \succ a_i \succ  b_i
  & \textrm{ for each $i \in [k]$;} \\
a_i\colon & b_i\succ  s_i
  & \textrm{ for each $i \in [k]$;} \\
b_i\colon & s_i\succ  a_i
  & \textrm{ for each $i \in [k]$.}
\end{array}
\]
The set of original agents is~$A \setminus V$, and $V$ is the set of addable agents.
We set $k$ as our budget, i.e., the number of agents that can be added to~$I$ from among~$V$.

We claim that $G$ contains an independent set of size~$k$ if and only if there are $k$ agents in~$V$ whose addition to~$I$ yields an instance that contains a stable matching; moreover, such a stable matching is necessarily perfect.

Assume first that there is an independent set~$X$ of size~$k$ in~$G$.
Let $I'$ be the SR instance obtained by adding the agents of~$X$ to~$I$. We show that $I'$ admits a stable and perfect matching. Let us match the agents of $S\coloneqq\{s_1,\dots,s_k\}$ to agents of~$X$ in a stable manner: Since the sub-instance of~$I$ induced by $S \cup X$ is an instance of \textsc{Stable Marriage} with complete preference lists, there exists a stable matching~$M_X$ that covers all  agents in~$S \cup X$ and admits no blocking pair formed by agents in~$S \cup X$.
Then the matching $M\coloneqq\{\{a_i,b_i\}\colon i \in[k]\} \cup M_X$  covers all agents in~$I'$ and is also stable in~$I'$.

Assume now that there is a set~$X$ of at most~$k$ agents in~$V$ whose addition to~$I$ yields an instance~$I'$ that admits a stable matching~$M$. Observe that for each $i \in [k]$, the matching $M$ must match $s_i$ to an agent in~$V$ due to the circular preferences of $a_i,b_i,$ and $s_i$. Thus, $X$ must contain exactly~$k$ agents from~$V$. It suffices now to note that $X$ forms an independent set in~$G$: assuming that two vertices of~$X$, say $x$ and $y$, are adjacent, the pair $\{x,y\}$ would block~$M$, a contradiction proving the correctness of our reduction.
}

Next, we prove that the polynomial-time algorithm by Boehmer et al.~\cite{boe-bre-hee-nie:j:control-bribery-in-SM} for \textsc{CSM}-$\DelAcc$-$\USM$ does not generalize to the roommates setting; we present a reduction from the $\np$-complete \textsc{Back and Forth Paths} problem~\cite{for-hop-wyl:j:back-and-forth-paths}.

\begin{restatable}[\linkproof{thm:CSR-DelAcc-USM}]{theorem}{thmCSRDelAccUSM}    
\label{thm:CSR-DelAcc-USM}
\textsc{CSR}-$\DelAcc$-$\USM$ is $\np$-complete, even on instances whose underlying graph can be made bipartite by deleting a single edge.     
\end{restatable}

\begin{proofsketch}%
Whether an instance of \textsc{Stable Roommates} admits a unique stable matching can be decided in polynomial time by standard methods~\cite{gus-irv:b:the-stable-marriage-problem}.
To show $\np$-hardness, we reduce from the following variant of the $\np$-complete problem \textsc{Back and Forth Paths}~\cite{for-hop-wyl:j:back-and-forth-paths}: %
given a directed graph~$D=(V,E)$ and two arcs~$(s,s')$ and~$(t,t')$ in~$D$, decide if there exists a cycle through $(s,s')$ and~$(t,t')$ in~$D$. We may assume that every vertex in~$D$ has indegree at least~$2$, and that $(s,v) \in E$ for each %
$v \in V \setminus \{s,s'\}$; see 
\ifshort
    the full version~\cite{che-sch:t:control-SM-arxiv} 
\else
\Cref{app:proof:thm:CSR-DelAcc-USM}  
\fi 
for details. 

For some~$v \in V$, let $\mydeg{v}$ denote the indegree of~$v$. We let $V'=V \setminus \{s',t'\}$. 
For each vertex~$v$, we fix an ordering of the incoming arcs of~$v$, and for each $i \in [\mydeg{v}]$ we denote by~$\delta(v,i)$ the tail of the $i$-th arc in this ordering. For simplicity, we let $(s,s')$ and $(t,t')$ be the first arcs among the incoming edges of~$s'$ and of~$t'$, respectively, so $\delta(s',1)=s$ and $\delta(t',1)=t$.

\mypara{Construction.}
We construct an instance of \textsc{Stable Roommates} as follows. 
For each vertex~$v \in V$, we create two sets of agents, $A_v=\{a_v^i \colon i \in [\mydeg{v}]\}$ and~${B_v=\{b_v^i \colon i \in [\mydeg{v}]\}}$. The preferences of these agents are as follows; throughout the proof, we treat the superscripts of agents in~$A_v \cup B_v$ modulo~$\mydeg{v}$: 
\[
\begin{array}{lll}
a_v^1 \colon & \{b_u^j: \delta(u,j)=v\} \succ b_v^{\mydeg{v}} \succ b_v^1
 \, &  \text{ for all } v\in V,\\[4pt]
a_v^i \colon & b_v^{i-1} \succ b_v^i
& \text{ for all } v\in V,\ 2\leq i\leq \mydeg{v},
\\[4pt]
b_v^i \colon & a_v^{i} \succ a_v^{i+1} \succ a_u^1 \qquad & \text{ for all $v \in V'$, $i \in [\mydeg{v}]$, $\delta(v,i)=u$,}
\\[4pt]
b_{s'}^1 \colon & a_{s'}^1 \succ a_{s'}^2 \succ b_{t'}^1 \succ a_s^1, \qquad &
\\[4pt]
b_{t'}^1 \colon & a_{t'}^1 \succ a_{t'}^2 \succ b_{s'}^1 \succ a_t^1, \qquad &
\\[4pt]
b_x^i \colon & a_x^{i} \succ a_x^{i+1} \succ a_u^1 \qquad & \text{ for $x \in \{s',t'\}$, $1<i \leq \mydeg{x}$, $\delta(x,i)=u$.}
\end{array}
\]

We set the target matching to $M=\{\{a_v^i,b_v^i\}:v \in V,i \in [\mydeg{v}]\}$ and the budget to~$\ell=|V|$.
Note that the instance~$I$ of \textsc{Stable Roommates} is almost bipartite: deleting the edge~$\{b_{s'}^1,b_{t'}^1\}$ yields an instance of \textsc{Stable Marriage}.

\mypara{Correctness.}
First assume that $F$ is a set of at most~$\ell$ edges in~$E(I)$ such that $M$ is the unique stable matching in~$I'=I-F$.
For each~$v \in V$, let $C_v$ denote the unique $M$-alternating cycle over agent set~$A_v \cup B_v$. It is not hard to see that $M_v =M \symdiff C_v$ (where $\symdiff$ stands for symmetric difference) is a stable matching in~$I$. Hence, $M_v$ is also a stable matching in~$I'$ unless $F$ contains an edge of~$C_v \setminus M$.
Due to~$|F| \leq |V|$, we know that $F$ contains exactly one edge from each~$C_v \setminus M$; let $\{b_v^{\sigma(v)},a_v^{\sigma(v)+1}\}$ denote this edge for some $\sigma(v) \in [\mydeg{v}]$.

For each~$v \in V$, consider the arc~$(u,v)$ in~$D$ such that $\delta(v,\sigma(v))=u$; we will use the notation
$\alpha(v)=u$. Define  $E^\star=\{(\alpha(v),v)\colon v \in V\}\subseteq E$. We are going to show that~$E^\star$ contains a cycle traversing both~$(s,s')$ and~$(t,t')$. 
Clearly, $E^\star$ contains some directed cycle~$C$, since each vertex has indegree exactly~$1$ in~$E^\star$. 
We define an $M$-alternating cycle~$\varphi(C)$ corresponding to~$C$ as the union of the paths $P_{(u,v)}=(a_v^1,b_v^1,\dots,b_v^{\sigma(v)},a_u^1)$ for $u=\alpha(v)$, taken over each $v \in V(C)$.
It can be shown that if $C$ contains at most one arc from~$\{(s,s'),(t,t')\}$, then the matching~$M'=M \symdiff \varphi(C)$ is stable in~$I'$, because each agent~$a_v^i$ for which $M'(a_v^i) \neq M(a_v^i)$ prefers~$M'$ to~$M$, and each agent~$b_v^i$ for which $M'(b_v^i) \neq M(b_v^i)$ is matched by~$M'$ to its second best choice in~$I'$, except possibly for either~$b_{s'}^1$ or~$b_{t'}^1$; thus, the only possible blocking edge for~$M'$ is $e^\star=\{b_{s'}^1,b_{t'}^1\}$, but by our assumption, at least one endpoint of~$e^\star$ prefers its partner both in~$M$ and in~$M'$ to the other endpoint of~$e^\star$.
Thus, $M'$ is stable in~$I'$, a contradiction. This implies that $C$ must be a cycle in~$D$ containing both~$(s,s')$ and~$(t,t')$.

Assume now that $C$ is a cycle in~$D$ containing both~$(s,s')$ and~$(t,t')$. Let $\alpha(v)$ be the vertex preceding~$v$ on~$C$ if $v \in V(C)$, and let $\alpha(v)=s$ otherwise; then $\{(\alpha(v),v)\colon v \in V\} \subseteq E$.
We define $F=\{\{b_v^{i},a_v^{i+1}\} \colon v \in V,\ \alpha(v)=\delta(v,i)\}$ and claim that $M$ is the unique stable matching in~$I'=I-F$. Since stability of~$M$ is clear, it remains to prove the uniqueness in~$I'$.

Assume for contradiction that $I'$ admits a stable matching other than~$M$. 
By the work of Cheng and Rosenbaum~\cite{che-ros:x:optimal-SM}, we know that there exists a stable matching~$M'$ in~$I'$---namely, a stable matching in~$I'$ that is adjacent to~$M$ in the median graph of the instance~$I'$---for which $M \symdiff M'$ is a single cycle~$C_{\Delta}$. Since deleting $\{b_{s'}^1,b_{t'}^1\}$ yields a bipartite instance, $C_{\Delta}$ cannot contain $\{b_{s'}^1,b_{t'}^1\}$.  

We say that some~$v \in V$ is \emph{important}, if $C_{\Delta}$ contains a vertex in~$A_v \cup B_v$. Note that $V(C_{\Delta}) \subseteq A_v \cup B_v$ is not possible, because $A_v \cup B_v$ induces a path in~$I'$. Thus, if $v$ is important, then $C_{\Delta}$ needs to contain at least two edges having exactly one endpoint in~$A_v \cup B_v$, and at most one such edge can be incident to~$a_v^1$. Thus, at least one edge on~$C_{\Delta}$ connects some~$b_v^{i}$ to~$a_u^1$ for some $i \in [\mydeg{v}]$ and~$u=\delta(v,i)$. %
However, then $\{b_v^i,a_v^{i+1}\}$ must be in~$F$, as otherwise it blocks~$M'$. This means that
$\{b_v^i,a_v^{i+1}\} \in F$ and, therefore,
$\alpha(v)=u=\delta(v,i)$.

Moreover, using the structure of our construction and the fact that $C_{\Delta}$ is $M$-alternating, it can be shown that the edges of~$C_{\Delta}$ incident to vertices of~$A_v \cup B_v$ are exactly the edges of~$P_{(u,v)}$ for $\alpha(v)=u$ and the edge~$\{a_v^1,b_z^j\} \in P_{(v,z)}$ where $\alpha(z)=v=\delta(z,j)$.
Consequently, $C_{\Delta}$ must be the union of all paths~$P_{(\alpha(v),v)}$ over all important vertices~$v$. It follows that the arcs~$(\alpha(v),v) \in E$ for each important vertex~$v$ form a single cycle in~$D$ which, thus, must be the cycle~$C$; in fact, using the notation~$\varphi(\cdot)$ for the first part of the proof, we have~$\varphi(C)=C_{\Delta}$. In particular, $\{b_{s'}^1,a_s^1\}$ and~$\{b_{t'}^1,a_t^1\}$ are both on~$C_{\Delta} \cap M'$. However, then $\{b_{s'}^1,b_{t'}^1\}$ forms a blocking pair for~$M'$, a contradiction to the stability of~$M'$. This proves that $M$ is indeed the unique stable matching in~$I'=I-F$, proving the correctness of our reduction. %
\end{proofsketch}

\prooftoappendix{thm:CSR-DelAcc-USM}{\thmCSRDelAccUSM*}{
Whether an instance of \textsc{Stable Roommates} admits a unique stable matching can be decided in polynomial time by standard methods~\cite{gus-irv:b:the-stable-marriage-problem}.
To show $\np$-hardness, we reduce from the $\np$-complete problem \textsc{Back and Forth Paths}~\cite{for-hop-wyl:j:back-and-forth-paths} which, given a directed graph~$D$ with two terminals, $s$ and~$t$, asks whether there exists a cycle through vertices~$s$ and~$t$ in~$D$.
Note that this problem reduces to the analogous problem where we ask for a cycle through two arcs, $(s,s')$ and~$(t,t')$, in a directed graph~$D=(V,E)$; this follows from the usual terminal version by splitting the two terminals and making the split arcs prescribed.
We may assume that $D$ contains no vertices without incoming arcs.
Moreover, we add an arc from~$s$ to each vertex in~$V \setminus \{s,s'\}$ (unless such an arc is already present in~$D$),
and if at this point there is still a vertex~$u$ with indegree less than~$2$, then we add an arc from~$t$ to~$u$. 
Note that such arcs cannot be contained in any solution (i.e., a cycle through~$(s,s')$ and~$(t,t')$), and thus adding them does not change the answer to our instance.

For some~$v \in V$, let $\mydeg{v}$ denote the indegree of~$v$. We let $V'=V \setminus \{s',t'\}$. We will denote the symmetric difference of two edge sets~$E_1$ and~$E_2$ in~$I$ as~$E_1 \symdiff E_2$. %
For each vertex~$v$, we order the incoming arcs of~$v$ arbitrarily, and for each $i \in [\mydeg{v}]$ we denote by~$\delta(v,i)$ the tail of the $i$-th arc in this ordering. For simplicity, we will assume that the arc~$(s,s')$ and $(t,t')$ are the first arcs among the incoming edges of~$s'$ and of~$t'$, respectively, so $\delta(s',1)=s$ and $\delta(t',1)=t$.

\mypara{Construction.}
We construct an instance of \textsc{Stable Roommates} as follows. 
For each vertex~$v \in V$, we create a set~$A_v=\{a_v^i:i \in [\mydeg{v}]\}$ and a set~$B_v=\{b_v^i:i \in [\mydeg{v}]\}$ of agents. The preferences of these agents are as follows; throughout the proof, we treat the superscripts of agents in~$A_v \cup B_v$ modulo~$\mydeg{v}$: 
\[
\begin{array}{lll}
a_v^1 \colon & \{b_u^j: \delta(u,j)=v\} \succ b_v^{\mydeg{v}} \succ b_v^1
 \, &  \text{ for all } v\in V,\\[4pt]
a_v^i \colon & b_v^{i-1} \succ b_v^i
& \text{ for all } v\in V,\ 2\leq i\leq \mydeg{v}, 
\\[4pt]
b_v^i \colon & a_v^{i} \succ a_v^{i+1} \succ a_u^1 \qquad & \text{ for all $v \in V'$, $i \in [\mydeg{v}]$, $\delta(v,i)=u$,}
\\[4pt]
b_{s'}^1 \colon & a_{s'}^1 \succ a_{s'}^2 \succ b_{t'}^1 \succ a_s^1, \qquad & 
\\[4pt]
b_{t'}^1 \colon & a_{t'}^1 \succ a_{t'}^2 \succ b_{s'}^1 \succ a_t^1, \qquad & 
\\[4pt]
b_x^i \colon & a_x^{i} \succ a_x^{i+1} \succ a_u^1 \qquad & \text{ for $x \in \{s',t'\}$, $1<i \leq \mydeg{x}$, $\delta(x,i)=u$.}
\end{array}
\]
We set the target matching to $M=\{\{a_v^i,b_v^i\}:v \in V,i \in [\mydeg{v}]\}$ and the budget to~$\ell=|V|$.
Note that the instance~$I$ of \textsc{Stable Roommates} is almost bipartite: deleting the edge~$\{b_{s'}^1,b_{t'}^1\}$ yields an instance of \textsc{Stable Marriage}.

\mypara{Correctness.}
First assume that $F$ is a set of at most~$\ell$ edges in~$E(I)$ such that $M$ is the unique stable matching in~$I'=I-F$. 
For each~$v \in V$, define the unique $M$-alternating cycle over agent set~$A_v \cup B_v$:
\[
C_v=(a_v^1,b_v^1,a_v^2,b_v^2,\ldots,a_v^{\mydeg{v}},b_v^{\mydeg{v}},a_v^1).
\]
Then
\[
C_v\setminus M=\{\{b_v^i,a_v^{i+1}\}:i\in[\mydeg{v}]\}.
\]
Observe that $M_v =M \symdiff C_v$ is a stable matching in~$I$, because  each agent~$a_v^i \in A_v$ prefers~$M_v(a_v^i)=b_v^{i-1}$ to its partner in~$M$, and each agent~$b_v^i \in B_v$ is matched to~$a_v^{i+1}$, its second best choice, with its best choice being $M(b_v^i)=a_v^i$. Hence, $M_v$ is also a  stable matching in~$I'$ unless $F$ contains an edge of~$C_v \setminus M$.
Due to~$|F| \leq |V|$, we know that $F$ contains exactly one edge from each~$C_v \setminus M$; let $\{b_v^{\sigma(v)}, a_v^{\sigma(v)+1}\}$ denote this edge for some $\sigma(v) \in [\mydeg{v}]$.

For each~$v \in V$, consider the arc~$(u,v)$ in~$D$ such that $\delta(v,\sigma(v))=u$; we will use the notation
$\alpha(v)=u$. Define  $E^\star=\{(\alpha(v),v)\colon v \in V\}\subseteq E$. We are going to show that~$E^\star$ contains a cycle traversing both~$(s,s')$ and~$(t,t')$. 
Clearly, $E^\star$ contains some directed cycle~$C$, since each vertex has indegree exactly~$1$ in~$E^\star$. 
We define an $M$-alternating cycle~$\varphi(C)$ corresponding to~$C$ as the union of the paths $P_{(u,v)}=(a_v^1,b_v^1,\dots,b_v^{\sigma(v)},a_u^1)$ for $u=\alpha(v)$, taken over each $v \in V(C)$.
If $C$ contains neither~$(s,s')$ nor~$(t,t')$, then  the matching~$M'=M \symdiff \varphi(C)$ is stable in~$I'$, because each agent~$a_v^i$ for which $M'(a_v^i) \neq M(a_v^i)$ prefers~$M'$ to~$M$, and each agent~$b_v^i$ for which $M'(b_v^i) \neq M(b_v^i)$ is matched by~$M'$ to its second best choice in~$I'$ (with the top choice being its partner in~$M$). If $C$ contains exactly one of~$(s,s')$ and~$(t,t')$, then for the same reasons, the only possible blocking edge for~$M'$ is $\{b_{s'}^1,b_{t'}^1\}$: however, since either~$\{b_{s'}^1,a_s^1\}$ or~$\{b_{t'}^1,a_t^1\}$ is not on~$\varphi(C)$, we know that one endpoint of~$e^\star=\{b_{s'}^1,b_{t'}^1\}$ prefers its partner both in~$M$ and in~$M'$ to the other endpoint of~$e^\star$.
Thus, $M'$ is stable in~$I'$, a contradiction. This implies that $C$ must be a cycle in~$D$ containing both~$(s,s')$ and~$(t,t')$.

For the other direction, assume now that $C$ is a cycle in~$D$ containing both~$(s,s')$ and~$(t,t')$. Let $\alpha(v)$ be the vertex preceding~$v$ on~$C$ if $v \in V(C)$, and let $\alpha(v)=s$ otherwise; then $\{(\alpha(v),v)\colon v \in V\} \subseteq E$.
We define a set $F$ of edges in~$E(I)$ as $F=\{\{b_v^{i},a_v^{i+1}\} \colon v \in V,\ \alpha(v)=\delta(v,i)\}$. We claim that $I'=I-F$ contains only~$M$ as a stable matching. First, $M$ is clearly stable in~$I'$; in fact, it is stable in~$I$. It remains to prove the uniqueness of~$M$ in~$I'$.

Assume for the sake of contradiction that $I'$ admits a stable matching other than~$M$. 
By the work of Cheng and Rosenbaum~\cite{che-ros:x:optimal-SM}, we know that there exists a stable matching~$M'$ in~$I'$---namely, a stable matching in~$I'$ that is adjacent to~$M$ in the median graph of the instance~$I'$---for which the symmetric difference of~$M$ and~$M'$ is a single cycle~$C_{\Delta}$.
Since deleting $\{b_{s'}^1,b_{t'}^1\}$ yields a bipartite instance with $\bigcup_{v \in V}A_v$ and $\bigcup_{v \in V}B_v$ forming the two sides of the bipartition, there is no path between~$b_{s'}^1$ and $b_{t'}^1$ avoiding the edge~$\{b_{s'}^1,b_{t'}^1\}$ that contains an odd number of edges. Since $C_{\Delta}$ is $M$-alternating and hence of even length, $C_{\Delta}$ cannot contain $\{b_{s'}^1,b_{t'}^1\}$.

We say that some~$v \in V$ is \emph{important}, if $C_{\Delta}$ contains a vertex in~$A_v \cup B_v$. Note that $V(C_{\Delta}) \subseteq A_v \cup B_v$ is not possible, because $A_v \cup B_v$ induces a path in~$I'$ (as $F$ contains exactly one edge from the $M$-alternating cycle induced by~$A_v \cup B_v$ in~$I$). Thus, if $v$ is important, then $C_{\Delta}$ needs to contain at least two edges having exactly one endpoint in~$A_v \cup B_v$; at most one such edge can be incident to~$a_v^1$, since the other edge of~$C_{\Delta}$ incident to~$a_v^1$ is the $M$-edge $\{a_v^1,b_v^1\}$, which lies inside~$A_v \cup B_v$. Hence, there must exist at least one edge on~$C_{\Delta}$ that connects some~$b_v^{i}$ to~$a_u^1$ for some $i \in [\mydeg{v}]$ and~$u=\delta(v,i)$. %
However, then $\{b_v^i,a_v^{i+1}\}$ must be in~$F$, as otherwise it blocks~$M'$: agent~$b_v^i$ prefers~$a_v^{i+1}$ to $M'(b_v^i)=a_u^1$ and $b_v^i$ is the top choice of~$a_v^{i+1}$. This means that
$\{b_v^i,a_v^{i+1}\} \in F$ and therefore
$\alpha(v)=u=\delta(v,i)$.

Moreover, since $C_{\Delta}$ is $M$-alternating, it must contain $\{a_v^i,b_v^i\} \in M$. In the case  $i>1$, the cycle $C_{\Delta}$ must also contain the only non-$M$ edge incident to~$a_v^i$, namely $\{a_v^i,b_v^{i-1}\}$, followed by~$\{a_v^{i-1},b_v^{i-1}\}$. Since this reasoning only breaks for~$i=1$, it follows that the edges of~$C_{\Delta}$ incident to vertices of~$A_v \cup B_v$ are exactly the edges of~$P_{(u,v)}$ for $\alpha(v)=u$, together with the edge~$\{a_v^1,b_z^j\} \in P_{(v,z)}$ where $\alpha(z)=v=\delta(z,j)$.
Consequently, $C_{\Delta}$ must be the union of all paths~$P_{(\alpha(v),v)}$ over all important vertices~$v$. It follows that the arcs~$(\alpha(v),v) \in E$ for each important vertex~$v$ form a single cycle in~$D$ which, thus, must be the cycle~$C$; in fact, using the notation~$\varphi(\cdot)$ for the first part of the proof, we have~$\varphi(C)=C_{\Delta}$. In particular, $\{b_{s'}^1,a_s^1\}$ and~$\{b_{t'}^1,a_t^1\}$ are both on~$C_{\Delta} \cap M'$. However, then $\{b_{s'}^1,b_{t'}^1\}$ forms a blocking pair for~$M'$, a contradiction to the stability of~$M'$. This proves that $M$ is indeed the unique stable matching in~$I'=I-F$, proving the correctness of our reduction. %
}

\section{Bribery actions}
\label{sec:bribery}
\appendixsection{sec:bribery}

In this section, we focus on the bribery actions that manipulate agents' preferences via the $\Swap$ or the $\Reorder$ actions. 
Within the section, we again organize our results based on the goal we aim to achieve.

\mypara{Matching an agent.}
Boehmer et al.~\cite{boe-bre-hee-nie:j:control-bribery-in-SM} showed that, already in the marriage setting, it is $\np$-complete to decide whether reordering at most $\ell$ preference lists yields a stable matching that contains a given pair~($\MP$).
Their hardness for~$\MP$ does not immediately imply hardness for~$\MA$:
a straightforward reduction from~$\MP$ to~$\MA$ would require one of the agents of the distinguished pair to have only the other as its acceptable partner, since this is the only way to make ``the distinguished agent is matched in some stable matching'' equivalent to ``the distinguished pair is contained in some stable matching''.
The construction of Boehmer et al.\ does not enjoy this property.
We therefore provide a direct reduction for~$\MA$ in the marriage setting.
Note that in our construction the distinguished agent~$a^*$ has a unique acceptable partner~$b^*$, so a stable matching covers~$a^*$ if and only if it contains the pair~$\{a^*,b^*\}$; moreover, the admissible operations for the goal~$\MA$ (which must avoid $a^*$ and $A(a^*)=\{b^*\}$) coincide with those for the goal~$\MP$ with distinguished pair~$\{a^*,b^*\}$.
Hence, our reduction also serves as a valid reduction for~$\MP$, and thus subsumes the hardness of \CSM-$\Reorder$-$\MP$ proved by Boehmer et al.

\begin{restatable}[\linkproof{thm:CSM-CSR-Reorder-MA}]{theorem}{thmCSMCSRReorderMA}
\label{thm:CSM-CSR-Reorder-MA}
\textsc{CSM}-$\Reorder$-$\MA$ and \textsc{CSR}-$\Reorder$-$\MA$ are $\np$-complete.
\end{restatable}
\prooftoappendix{thm:CSM-CSR-Reorder-MA}{\thmCSMCSRReorderMA*}{
  Since \textsc{Stable Marriage} is a special case of \textsc{Stable Roommates}, it suffices to establish $\np$-hardness for the marriage case and $\np$-membership for the roommates case.
  Membership in~$\np$ follows by the same argument as in \cref{thm:CSM-AddAg-MA}:
  given a set of reorder operations resulting in a modified instance~$I'$, the Rural Hospitals Theorem~\cite{irv-lea-gus:j:efficient-algorithm-for-optimal-stable-marriage} implies that every stable matching in~$I'$ covers the same set of agents;
  hence, it suffices to compute a single stable matching in~$I'$ (which can be done in polynomial time~\cite{irv:j:efficient-algorithm-for-stable-roommates-problem}) and check whether the distinguished agent is matched.

  To show $\np$-hardness for the marriage case, we give a polynomial-time many-to-one reduction from the $\np$-complete \textsc{Clique} problem~\cite{gar-joh:b:int}; the problem definition can be found before \cref{thm:CSM-AddAg-MA}.
  Let $G=(V,E)$ be the input graph and~$k$ the size of the desired clique; without loss of generality, we assume $k\ge 4$.
  We construct an instance of \textsc{Stable Marriage} with the following agents.
  \begin{compactitem}
  \item One \myemph{distinguished man}~$a^*$ and one \myemph{distinguished woman}~$b^*$; our goal is to match~$a^*$ in a stable matching.
  \item For each vertex~$v\in V$, two \myemph{vertex men}~$a_v$ and~$c_v$, and one \myemph{vertex woman}~$b_v$;
    let $A\coloneqq\{a_v\colon v\in V\}$, $B\coloneqq\{b_v\colon v\in V\}$, and $C\coloneqq\{c_v\colon v\in V\}$.
  \item For each edge~$e\in E$, one \myemph{edge man}~$q_e$ and two \myemph{edge women}~$p_e$ and~$r_e$;
    let $P\coloneqq\{p_e\colon e\in E\}$, $Q\coloneqq\{q_e\colon e\in E\}$, and $R\coloneqq\{r_e\colon e\in E\}$.
  \item $\binom{k}{2}$ \myemph{selector men}~$s_1, s_2,\dots, s_{\binom{k}{2}}$; let $S\coloneqq\{s_z\colon z\in [\binom{k}{2}]\}$.
  \end{compactitem}
  Thus, the set of men is $\{a^*\}\cup A \cup C \cup Q \cup S$ and the set of women is $\{b^*\}\cup B \cup P \cup R$.
  The preferences of the agents are as follows (with each set in the preference lists interpreted as an arbitrary but fixed, strict linear order of its elements):
\[
\begin{array}{llllll}
  a^*\colon & b^*; & & b^*\colon & S\succ Q\succ a^*; & \\
  s_z\colon & P\succ b^*, & \forall z\in[\tbinom{k}{2}]; \quad & r_e\colon & q_e, & \forall e\in E; \\
  q_e\colon & p_e\succ b^*\succ r_e, & \forall e\in E; & p_e\colon & q_e\succ a_v\succ a_{v'}\succ S,\,\, & \forall e=\{v,v'\}\in E; \\
  a_v\colon & b_v\succ\{p_e\colon v\in e\},\,\, & \forall v\in V; & b_v\colon & c_v\succ a_v, & \forall v\in V; \\
  c_v\colon & b_v, & \forall v\in V. & & & \\
\end{array}
\]
  Finally, we set the budget to $\ell\coloneqq\binom{k}{2}+k$.
  Note that since $b^*$ is the only acceptable partner of~$a^*$, the admissibility restriction for the goal~$\MA$ forbids reordering the preference lists of both~$a^*$ and~$b^*$.

  It remains to prove correctness: $G$ admits a clique of size~$k$ if and only if reordering at most~$\ell$ admissible preference lists yields an instance with a stable matching that covers~$a^*$.

  For the forward direction, assume that $G$ admits a clique~$K$ of size~$k$;
  let $E'\subseteq E$ and $V'\subseteq V$ denote its edge set and its vertex set, respectively, so that $|E'|=\binom{k}{2}$ and $|V'|=k$.
  We reorder the preference lists of the $\binom{k}{2}+k=\ell$ agents in $\{q_e\colon e\in E'\}\cup\{b_v\colon v\in V'\}$ as follows:
  \[
    q_e\colon r_e\succ\cdots, \;\;\forall e\in E';
    \qquad
    b_v\colon a_v\succ c_v, \;\;\forall v\in V'.
  \]
  Let $I'$ denote the resulting instance.
  We construct a matching~$M$ in~$I'$ as follows:
  \begin{compactitem}
    \item match $a^*$ with $b^*$;
    \item match the selector men in~$S$ with the edge women in $\{p_e\colon e\in E'\}$ in a stable way (see footnote~\ref{fn:stableway});
    \item for each clique edge~$e\in E'$, match $q_e$ with $r_e$;
    \item for each non-clique edge~$e\in E\setminus E'$, match $q_e$ with $p_e$;
    \item for each clique vertex~$v\in V'$, match $b_v$ with $a_v$;
    \item for each non-clique vertex~$v\in V\setminus V'$, match $b_v$ with $c_v$.
  \end{compactitem}
  Clearly, $M$ covers~$a^*$; the only agents left unmatched are those in $\{a_v\colon v\in V\setminus V'\}$ and $\{r_e\colon e\in E\setminus E'\}$.

  We now verify that $M$ is stable in~$I'$.
  \begin{compactitem}
    \item No blocking pair contains~$a^*$, since he is matched with his only acceptable partner.
    \item No blocking pair contains a selector man~$s_z\in S$ and~$b^*$, since $s_z$ is matched with an edge woman in~$P$ whom he prefers to~$b^*$.
    \item No blocking pair contains a selector man~$s_z\in S$ and an edge woman~$p_e$:
    if $e\notin E'$, then $p_e$ is matched with her most preferred man~$q_e$;
    otherwise $M(p_e)\in S$, so $\{s_z,p_e\}$ does not block~$M$ by the stability of the restriction of~$M$ to~$S\cup\{p_e\colon e\in E'\}$.
    \item No blocking pair contains an edge man~$q_e$, since $q_e$ is matched with his most preferred woman in~$I'$:
    if $e\in E'$, then the preference list of~$q_e$ has been reordered and $q_e$ is matched with his new top choice~$r_e$;
    otherwise the preference list of~$q_e$ is unchanged and $q_e$ is matched with his original top choice~$p_e$.
    \item Similarly, no blocking pair contains a vertex woman~$b_v$, since $b_v$ is matched with her most preferred man in~$I'$:
    if $v\in V'$, then the preference list of~$b_v$ has been reordered and $b_v$ is matched with her new top choice~$a_v$;
    otherwise the preference list of~$b_v$ is unchanged and $b_v$ is matched with her original top choice~$c_v$.
    \item Consequently, no blocking pair involves~$r_e$ (resp., $c_v$), since $r_e$ (resp., $c_v$) only finds $q_e$ (resp., $b_v$) acceptable and no blocking pair involves~$q_e$ (resp., $b_v$).
    \item Finally, no blocking pair contains a vertex man~$a_v$ and an edge woman~$p_e$ with $v\in e$:
    if $v\in V'$, then $a_v$ is matched with his most preferred woman~$b_v$;
    otherwise $v\in V\setminus V'$, which together with $v\in e$ implies $e\notin E'$, so $p_e$ is matched with her most preferred man~$q_e$.
  \end{compactitem}
  Hence $M$ is stable in~$I'$.

  Conversely, assume that reordering at most~$\ell$ admissible preference lists yields an instance~$I'$ that admits a stable matching~$M$ covering~$a^*$.
  Since $b^*$ is the only acceptable partner of~$a^*$, we have $\{a^*,b^*\}\in M$.
  Recall that the admissibility restriction forbids reordering~$b^*$'s preference list, which therefore coincides in~$I'$ with the original one.
  By the stability of~$M$, every agent~$x$ that $b^*$ prefers to~$a^*$ must be matched in~$M$ with some agent that $x$ prefers to~$b^*$.
  Examining $b^*$'s preferences, this implies that every selector man in~$S$ is matched with some edge woman in~$P$, and every edge man in~$Q$ is matched in~$M$.
  Consequently, exactly $\binom{k}{2}$ edges $e\in E$ satisfy $M(p_e)\in S$.

  We first show that for each such edge~$e$, the preference list of~$q_e$ must have been reordered in~$I'$.
  Indeed, in the original instance $p_e$ and $q_e$ are each other's most preferred agents, and $q_e$ prefers $b^*$ to $r_e$.
  Since $M(p_e)\in S$ and ${M(b^*)=a^*}$, if the preference list of~$q_e$ were unchanged, then $\{q_e,b^*\}$ would block~$M$ regardless of whether $q_e$ is matched (in which case $q_e$ must be matched with~$r_e$) or unmatched.
  Consequently, the preference lists of at least $\binom{k}{2}$ edge men in~$Q$ have been reordered in~$I'$, leaving a budget of at most $k$ for further reorderings.

  Define
  \[
    E^*\coloneqq\{e\in E\colon M(p_e)\in S\text{ and $p_e$'s preference list is unchanged in }I'\},
  \]
  and let $V^*\coloneqq\{v\in V\colon v\in e\text{ for some }e\in E^*\}$.
  We will show that $|V^*|=k$ and $|E^*|=\binom{k}{2}$, so that $V^*$ is the vertex set of a clique of size~$k$ in~$G$.

  We claim that for each vertex $v\in V^*$, at least one agent in the set $\{a_v,b_v\}\cup\{q_{e''}\colon {e''\ni v,}\,M(p_{e''})\notin S\}$ has been reordered in~$I'$.
  Indeed, fix any $e\in E^*$ with $v\in e$.
  Since $p_e$ is unchanged, $p_e$ prefers $a_v$ to her current partner $M(p_e)\in S$, so stability requires $a_v$ to be matched with some agent he prefers to~$p_e$ in~$I'$.
  If $a_v$'s preference list is reordered, the claim holds.
  Otherwise, the agents that~$a_v$ prefers to~$p_e$ are $b_v$ and a (possibly empty) subset of $\{p_{e''}\colon v\in e''\}$.
  If $M(a_v)=b_v$, then $b_v$ must be reordered:
  otherwise $b_v$ originally prefers~$c_v$ to~$a_v$, and since $c_v$ only finds~$b_v$ acceptable, the pair $\{b_v,c_v\}$ would block~$M$.
  Otherwise $M(a_v)=p_{e''}$ for some $e''\ni v$;
  since $p_{e''}$ is matched with~$a_v$, we have $M(p_{e''})\notin S$,
  and $q_{e''}$ must be matched, but neither with~$p_{e''}$ (taken by~$a_v$) nor with~$b^*$ (taken by~$a^*$), so $M(q_{e''})=r_{e''}$.
  Since $q_{e''}$ originally prefers~$b^*$ to~$r_{e''}$, the pair $\{q_{e''},b^*\}$ would block~$M$ unless $q_{e''}$ is reordered.

  Moreover, the reorderings provided by the claim are pairwise distinct across vertices:
  $a_v$ and $b_v$ are unique to~$v$, and a given $q_{e''}$ can be provided only for the unique vertex~$v'$ with $M(p_{e''})=a_{v'}$.
  Summing up, $I'$ contains the following three pairwise disjoint groups of reorderings:
  \begin{compactenum}[(a)]
    \item the $\binom{k}{2}$ reorderings of $\{q_e\colon M(p_e)\in S\}$ identified above;
    \item the reorderings of $p_e$ for the edges in $\{e\colon M(p_e)\in S\}\setminus E^*$, of which there are $\binom{k}{2}-|E^*|$ by the definition of~$E^*$;
    \item the $|V^*|$ per-vertex reorderings from the claim, which are disjoint from group~(a) (since each $q_{e''}$ above satisfies $M(p_{e''})\notin S$) and from group~(b) (which involves agents from~$P$  rather than agents from $A \cup B \cup Q$).
  \end{compactenum}
  The budget constraint thus yields
  \[
    \underbrace{\binom{k}{2}}_{\text{(a)}}+\underbrace{\bigl(\tbinom{k}{2}-|E^*|\bigr)}_{\text{(b)}}+\underbrace{|V^*|}_{\text{(c)}}\le\ell=\binom{k}{2}+k,
    \quad\text{i.e.,}\quad
    |V^*|+\bigl(\tbinom{k}{2}-|E^*|\bigr)\le k.
  \]
  Since every edge of~$E^*$ has both endpoints in~$V^*$, we have $|E^*|\le\binom{|V^*|}{2}$;
  combined with the previous inequality, $|V^*|+\binom{k}{2}-\binom{|V^*|}{2}\le k$, which rearranges to $|V^*|(|V^*|-3)\ge k(k-3)$.
  Since $k\ge 4$ and the function $x\mapsto x(x-3)$ is strictly increasing on $[3,\infty)$, this gives $|V^*|\ge k$;
  together with $|V^*|\le k$ (from the budget inequality and $|E^*|\le\binom{k}{2}$), we conclude $|V^*|=k$ and $|E^*|=\binom{k}{2}$.
  Hence $V^*$ is the vertex set of a clique of size~$k$ in~$G$, as desired.
}

The same construction with a larger budget yields hardness for $\Swap$.

\begin{restatable}[\linkproof{prop:CSM-CSR-Swap-MA}]{proposition}{propCSMCSRSwapMA}
\label{prop:CSM-CSR-Swap-MA}
\textsc{CSM}-$\Swap$-$\MA$ and \textsc{CSR}-$\Swap$-$\MA$ are $\np$-complete.
\end{restatable}
\prooftoappendix{prop:CSM-CSR-Swap-MA}{\propCSMCSRSwapMA*}{
Membership in~$\np$ follows as in \cref{thm:CSM-CSR-Reorder-MA}.
For $\np$-hardness, we use the same reduction from \textsc{Clique} as in \cref{thm:CSM-CSR-Reorder-MA}, with the same construction and agents/preferences, but set the budget to $\ell\coloneqq 2\binom{k}{2}+k$.

For the forward direction, given a clique~$K$ in~$G$ with edge set~$E'$ and vertex set~$V'$, we perform $2\binom{k}{2}+k$ swaps:
for each $e\in E'$, two swaps on~$q_e$'s preference list (first swap~$b^*$ and~$r_e$, then $p_e$ and~$r_e$), yielding $q_e\colon r_e\succ p_e\succ b^*$;
and for each $v\in V'$, one swap on~$b_v$'s preference list (swap~$c_v$ and~$a_v$), yielding $b_v\colon a_v\succ c_v$.
Stability of the matching~$M$ defined in the proof of \cref{thm:CSM-CSR-Reorder-MA} follows by the same case-by-case verification: each~$q_e$ (resp., $b_v$) is matched with its new top choice in~$I'$.

For the backward direction, the argument follows the proof of \cref{thm:CSM-CSR-Reorder-MA} with one adjustment to the swap accounting.
As before, every selector man is matched with some edge woman in~$P$, every edge man in~$Q$ is matched, and exactly $\binom{k}{2}$ edges $e\in E$ satisfy $M(p_e)\in S$.
For each such edge~$e$, the pair~$\{q_e,b^*\}$ would block~$M$ unless~$q_e$'s preference list is modified, so at least one swap on~$q_e$ is needed (specifically, swapping $b^*$ and~$r_e$ down); after this swap, $q_e\colon p_e\succ r_e\succ b^*$, and the pair~$\{q_e,p_e\}$ would block~$M$ unless either a second swap is performed on~$q_e$ (swapping $p_e$ and~$r_e$) or at least three swaps are performed on~$p_e$ (to move some selector man above~$q_e$, past~$a_v$, $a_{v'}$, and~$q_e$ in~$p_e$'s list).
Defining $E^*$ and~$V^*$ as in \cref{thm:CSM-CSR-Reorder-MA}, this yields the following lower bounds on the total swap count for edges in $\{e\colon M(p_e)\in S\}$:
\begin{compactitem}
\item for each $e\in E^*$ (i.e., $p_e$ unchanged), two swaps must lie on~$q_e$, contributing $2|E^*|$ swaps;
\item for each $e\in\{e\colon M(p_e)\in S\}\setminus E^*$ (i.e., $p_e$ has at least one swap), the combined swap count on $q_e$ and~$p_e$ is at least~$3$ (two swaps on~$q_e$ plus at least one on~$p_e$, or one on~$q_e$ plus at least three on~$p_e$), contributing $3(\binom{k}{2}-|E^*|)$ swaps.
\end{compactitem}
Together with the $|V^*|$ pairwise distinct per-vertex swaps from the claim in the proof of \cref{thm:CSM-CSR-Reorder-MA} (whose validity is unaffected by the change of action), the budget constraint becomes
\[ 
  2|E^*|+3\bigl(\tbinom{k}{2}-|E^*|\bigr)+|V^*|\le\ell=2\binom{k}{2}+k,
  \quad\text{i.e.,}\quad
  |V^*|+\bigl(\tbinom{k}{2}-|E^*|\bigr)\le k,
\]
which is the same inequality as in the proof of \cref{thm:CSM-CSR-Reorder-MA}.
The remainder of the argument is identical, yielding $|V^*|=k$ and $|E^*|=\binom{k}{2}$, so $V^*$ is a clique of size~$k$ in~$G$.
}

\begin{remark}
The $\np$-hardness for the swap action in the marriage case can alternatively be derived from \cref{thm:CSM-AddAg-MA} via a generic reduction from $\AddAg$ to $\Swap$ due to Boehmer et al.~\cite{boe-bre-hee-nie:j:control-bribery-in-SM}: for each addable agent~$p$ in the source instance, introduce two dummy agents~$p_d$ and~$q_d$ with preferences $p_d\colon p\succ q_d$ and $q_d\colon p_d$, and prepend~$p_d$ to the head of~$p$'s preference list; in addition, pad every other preference list with sufficiently many further dummies between consecutive entries so that no swap on any other list fits within the budget.
A single swap on~$p_d$'s list (from $p\succ q_d$ to $q_d\succ p$) causes $\{p_d,q_d\}$ to be matched, releasing~$p$ to participate in the rest of the matching---corresponding exactly to adding~$p$ in the source instance, with the budgets coinciding.
\end{remark}

\mypara{Stabilizing a matching.}
Next, we consider the goal~$\SM$.
Boehmer et al.~\cite{boe-bre-hee-nie:j:control-bribery-in-SM} show that both bribery actions~$\Swap$ and $\Reorder$ are polynomial-time solvable for the marriage case.
We complement this result by showing that for the roommates setting, both \CSR-$\Swap$-$\SM$ and \CSR-$\Reorder$-$\SM$ become $\np$-complete.

\begin{restatable}[\linkproof{thm:CSR-Swap-MS}]{theorem}{thmCSRSwapMS}
\label{thm:CSR-Swap-MS}
  \CSR-$\Swap$-$\SM$ is $\np$-complete.
\end{restatable}
\begin{proof}
  Membership in~$\np$ is straightforward, since checking whether a given matching is stable can be done in polynomial time.
  To show $\np$-hardness, we reduce from the following $\np$-complete variant~\cite{BermanKarpinskiScott2003} of the \TSAT\ problem which we call \RSAT:
  given a 3CNF Boolean formula~$\phi$ (i.e., a set of clauses each containing exactly $3$ literals) over a variable set~$X$, such that each variable appears \emph{exactly} twice unnegated and twice negated, decide if  $\phi$ is satisfiable.  

  Let $(X,\phi)$ be an instance of~\RSAT\ with $X=\{x_1,\dots,x_{\nn}\}$ and $\phi=\{C_1,\dots,C_{\mm}\}$.
  For each clause~$C_j$, we fix an arbitrary order of its three literals and write $C_j=(\lit_j^1\vee \lit_j^2\vee \lit_j^3)$.
  We construct an instance of \CSR-$\Swap$-$\SM$ as follows:
  for each variable~$x_i\in X$, we introduce four \myemph{variable agents}~$\xii,\nxi,\yi,$ and $\nyi$;
  for each clause~$C_j\in\phi$ and each $z\in[3]$, we introduce two \myemph{clause agents}~$c_j^z$ and~$d_j^z$.
  The preferences are as follows (indices for clause agents taken modulo~$3$): 
  \[\allowdisplaybreaks
    \begin{array}{lll}
      \xii\colon \nxi \succ \{c_j^z\colon \lit_j^z=\xii\}\succ \yi, \qquad& \yi\colon \xii, \qquad & \forall i\in[\nn]; \\[2pt]
      \nxi\colon \xii \succ \{c_j^z\colon \lit_j^z=\nxi\}\succ \nyi, & \nyi\colon \nxi, & \forall i\in[\nn]; \\[2pt]
      c_j^z\colon c_j^{z+1}\succ c_j^{z+2}\succ \lit_j^z\succ d_j^z, & d_j^z\colon c_j^z, & \forall j\in[\mm],\,z\in[3].
    \end{array}
  \]
  Since each literal occurs in exactly two clauses, each set $\{c_j^z\colon \lit_j^z=\lit\}$ has size~$2$, so the variable agents' preference lists have exactly four entries each.
  The target matching is
  $M\coloneqq\{\{\xii,\yi\},\{\nxi,\nyi\}\colon i\in[\nn]\}\cup\{\{c_j^z,d_j^z\}\colon j\in[\mm],\,z\in[3]\}$,
  and we set the budget to $\ell\coloneqq 3\nn+5\mm$.

 \appendixcorrectnessproofwithstatement{ 
 We now give some intuition on why $(X,\phi)$ is satisfiable if and only if at most~$\ell$ swaps suffice to obtain an instance with a stable matching~$M'\subseteq M$.
 First, for each $i \in [\nn]$ we need to ensure that $\{\xii,\nxi\}$ is not a blocking pair, which requires three swaps in the preferences of either~$\xii$ or~$\nxi$; this choice corresponds to setting $\xii$ to true or to false. 
 Next, for each $j \in [\mm]$ we need to ensure that no pair~$\{c_j^a,c_j^b\}$ for $a,b \in [3]$ is a blocking pair; this requires five swaps and, importantly, leaves the preferences of~$c_j^z$ for some $z \in [3]$ unchanged. Since these swaps already use up the budget, for each agent~$c_j^z$ with untouched preferences we must have that the corresponding variable agent's preference list  must have been changed (so that $\{c_j^z,l_j^z\}$ is not blocking)---this requirement translates to each clause containing a true literal. 
 For the complete proof of correctness, see 
 \ifshort 
    the full version~\cite{che-sch:t:control-SM-arxiv}.
 \else  
     \Cref{app:proof:thm:CSR-Swap-MS}.
 \fi }{thm:CSR-Swap-MS}
{\thmCSRSwapMS*}{
 \mypara{Correctness proof.}
  Before turning to the equivalence, we record a useful observation:
  For each modified instance and every stable matching~$M'$, if $M'\subseteq M$, then $M' = M$.
  Indeed, each of~$\yi$, $\nyi$, and~$d_j^z$ has only one acceptable partner (namely $\xii$, $\nxi$, and $c_j^z$, respectively), so omitting a pair of~$M$ from~$M'$ would leave both of its endpoints unmatched, forming a blocking pair.

  For the forward direction, assume that $\sigma\colon X\to\{\tru,\fal\}$ satisfies~$\phi$.
  For each ${i\in[\nn]}$ with $\sigma(\xii)=\tru$, we perform three swaps on~$\xii$'s preference list to move~$\yi$ from the last to the first position, resulting in $\xii\colon \yi\succ \nxi\succ\cdots$;
  for each~$i\in[\nn]$ with $\sigma(\xii)=\fal$, we perform the analogous three swaps on~$\nxi$'s list.
  This uses $3\nn$ swaps.
  Next, for each clause~$C_j$, we fix a literal $\lit_j^{\zstar}$ satisfied by~$\sigma$ and perform:
  \begin{compactitem}
    \item two swaps on $c_j^{\zstar+1}$'s list to move $d_j^{\zstar+1}$ to the second position, yielding $c_j^{\zstar+1}\colon c_j^{\zstar+2}\succ d_j^{\zstar+1}\succ c_j^{\zstar}\succ \lit_j^{\zstar+1}$;
    \item three swaps on $c_j^{\zstar+2}$'s list to move $d_j^{\zstar+2}$ to the first position, yielding $c_j^{\zstar+2}\colon d_j^{\zstar+2}\succ c_j^{\zstar}\succ c_j^{\zstar+1}\succ \lit_j^{\zstar+2}$.
  \end{compactitem}
  The list of $c_j^{\zstar}$ is left unchanged.
  This uses $5\mm$ further swaps, for a total of~$\ell$.

  Let~$I'$ denote the resulting instance.
  We verify that $M$ is stable in~$I'$ by considering each acceptable pair $\{u,v\}\notin M$:
  \begin{compactitem}
    \item $\{\xii,\nxi\}$: We only show the case when $\sigma(\xii)=\tru$ since the other case is analogous.
    Since $\sigma(\xii) = \tru$, by the swapping, $\yi$ is at the top of $\xii$'s list in~$I'$, so $\xii$ prefers $M(\xii)=\yi$ to~$\nxi$ and the pair does not block.
  \item $\{\xii,c_j^z\}$ with $\lit_j^z=\xii$: if $\sigma(\xii)=\tru$, then the previous argument gives that $\xii$ prefers~$\yi$ to~$c_j^z$;
    if $\sigma(\xii)=\fal$, then $\xii$ does not satisfy~$C_j$, so $z\neq \zstar$, and hence $c_j^z$'s list was modified placing $d_j^z$ above $\lit_j^z=\xii$; thus $c_j^z$ prefers $M(c_j^z)=d_j^z$ to~$\xii$.
    \item $\{\nxi,c_j^z\}$ with $\lit_j^z=\nxi$:
    The argument is analogous to the one for $\{\xii,c_j^z\}$.
  \item $\{c_j^a,c_j^b\}$ with $a\neq b$: by inspection of the three modified lists, at least one of $c_j^a,c_j^b$ has its partner in~$M$ ranked above the other (concretely: if $\zstar\in\{a,b\}$, then the modified one prefers its $d$-partner; if $\zstar\notin\{a,b\}$, then $c_j^{\zstar+2}$ has $d_j^{\zstar+2}$ at the top of its list).
  \end{compactitem}
  No blocking pair exists, so $M$ is stable in~$I'$.

  Conversely, suppose that at most~$\ell$ swaps applied to the constructed instance yield an instance~$I'$ in which some~$M'\subseteq M$ is stable; we know $M'=M$
  by the observation above.

  For each~$i\in[\nn]$, the pair $\{\xii,\nxi\}$ blocks~$M$ in the original instance (both rank the other above their partner in~$M$), so in~$I'$ either~$\yi$ must rank above~$\nxi$ in~$\xii$'s list or~$\nyi$ must rank above~$\xii$ in~$\nxi$'s list.
  Since $\yi$ originally sits at the last position of $\xii$'s four-agent list, the former requires at least three swaps on $\xii$'s list; the latter analogously requires at least three swaps on $\nxi$'s list.
  Hence at least three swaps occur on at least one of the two lists.
  These swaps are pairwise disjoint across variables and from the clause agents' lists, so the variable side contributes at least $3\nn$ swaps overall.
  We define a truth assignment~$\sigma\colon X\to\{\tru,\fal\}$ by setting $\sigma(\xii)=\tru$ if $\yi$ has been moved above~$\nxi$ in $\xii$'s list, and $\sigma(\xii)=\fal$ otherwise (in which case $\nyi$ has been moved above $\xii$ in~$\nxi$'s list).

  Next, fix~$j\in[\mm]$ and let $s_z\ge 0$ denote the number of swaps on~$c_j^z$'s list, for $z\in[3]$.
  The three pairs $\{c_j^a,c_j^b\}$ with $a\neq b$ all block~$M$ in the original instance.
  For~$\{c_j^a,c_j^b\}$ to not block in~$I'$, either~$c_j^a$'s list places $d_j^a$ above $c_j^b$ or~$c_j^b$'s list places $d_j^b$ above $c_j^a$.
  Since $c_j^b$ sits at position~$1$ of $c_j^a$'s list when $b=a+1$ and at position~$2$ when $b=a+2$ (with $d_j^a$ initially at position~$4$), this gives the constraints:
  \[
    (s_1\ge 3 \vee s_2\ge 2),\quad (s_3\ge 3 \vee s_1\ge 2),\quad (s_2\ge 3 \vee s_3\ge 2).
  \]
  A case analysis shows that $s_1+s_2+s_3\ge 5$, with equality only when exactly one of $s_1,s_2,s_3$ equals~$0$, and the other two equal~$2$ and~$3$ in the cyclic pattern $(s_{\zstar},s_{\zstar+1},s_{\zstar+2})=(0,2,3)$ for some $\zstar\in[3]$.
  These swaps are again pairwise disjoint across clauses and from the variable-side swaps, contributing at least $5\mm$ swaps overall.

  Combined with the variable-side swaps, the total of at least $3\nn+5\mm=\ell$ is met with equality.
  Hence each clause~$C_j$ uses exactly the pattern above, with some~$\zstar=\zstar(j)\in[3]$ such that~$c_j^{\zstar}$'s list is left unchanged.
  In~$I'$, $c_j^{\zstar}$ is matched with~$d_j^{\zstar}$ but prefers~$\lit_j^{\zstar}$ (at position~$3$) to~$d_j^{\zstar}$ (at position~$4$), so $\{c_j^{\zstar},\lit_j^{\zstar}\}$ would block~$M$ unless the variable agent~$\lit_j^{\zstar}$ in~$I'$ prefers its partner in~$M$ (a $y$-agent) to~$c_j^{\zstar}$.
  In the original list of $\lit_j^{\zstar}$, $c_j^{\zstar}$ lies above the $y$-agent, so this requires $\lit_j^{\zstar}$'s preference list to have been reordered;
  since the per-variable budget is tight, the three variable-side swaps must lie on $\lit_j^{\zstar}$'s list, which by definition of~$\sigma$ means $\sigma(\lit_j^{\zstar})=\tru$, i.e., $\lit_j^{\zstar}$ is satisfied.
  Thus~$\sigma$ satisfies every clause of~$\phi$.
}
\end{proof}

For the reorder action, the same hardness can be obtained via a reduction from \textsc{Vertex Cover} that re-uses part of the construction from the proof of \cref{thm:CSR-AddAg-SM}.

\begin{restatable}[\linkproof{prop:CSR-Reorder-MS}]{proposition}{propCSRReorderMS}
\label{prop:CSR-Reorder-MS}
  \CSR-$\Reorder$-$\SM$ is $\np$-complete.
\end{restatable}
\prooftoappendix{prop:CSR-Reorder-MS}{\propCSRReorderMS*}{
  Membership in~$\np$ is straightforward.
  For $\np$-hardness, we reduce from the $\np$-complete \textsc{Vertex Cover} problem~\cite{gar-joh:b:int}, which asks whether a given graph~$G=(V,E)$ admits a \myemph{vertex cover} of size at most~$k$, i.e., a set~$C\subseteq V$ such that every edge has at least one endpoint in~$C$.

  Given an instance~$(G,k)$, we restrict the construction in the proof of \cref{thm:CSR-AddAg-SM} to the agents $\{a_v,a'_v\colon v\in V\}$ with their original preferences:
  \[
    a_v\colon a'_v, \qquad a'_v\colon \{a'_u\colon \{u,v\}\in E\}\succ a_v, \qquad \forall v\in V.
  \]
  We define $M\coloneqq\{\{a_v,a'_v\}\colon v\in V\}$ as the target matching and set the budget to~$\ell\coloneqq k$.

  For the forward direction, given a vertex cover~$C\subseteq V$ of size at most~$k$, we reorder $a'_v$'s preference list for each $v\in C$ to place~$a_v$ at the top.
  Stability of~$M$ in the resulting instance follows by exactly the case analysis given in the proof of \cref{thm:CSR-AddAg-SM} (with $X\coloneqq V\setminus C$ playing the role of the independent set there):
  any candidate blocking pair has the form $\{a'_u,a'_v\}$ with $\{u,v\}\in E$, and since~$C$ covers~$\{u,v\}$, at least one of $a'_u,a'_v$ is matched with its new top choice and cannot block.

  Conversely, suppose that a set of at most~$\ell$ reorder operations yields an instance~$I'$ in which some $M'\subseteq M$ is stable. Let $R$ denote the agents whose preferences are reordered.
  Since each~$a_v$ has only one acceptable partner, any reordering of the preferences of~$a_v$ is vacuous;
  thus we may assume ${R\subseteq\{a'_v\colon v\in V\}}$.
  Since any pair $\{a_v,a'_v\}\notin M'$ would leave both endpoints unmatched and hence block~$M'$, we know $M'=M$.
  Define $C\coloneqq\{v\in V\colon a'_v\in R\}$, so $|C|\le k$.
  For any edge $\{u,v\}\in E$ with $u,v\notin C$, both~$a'_u$ and~$a'_v$ retain their original preferences, 
  in which case $\{a'_u,a'_v\}$ blocks~$M$ in~$I'$, a contradiction.
  Hence, $C$ is a vertex cover of size at most~$k$.
}

\mypara{Ensuring the existence of a stable matching.} %
 Finally, we consider the bribery counterpart of the existence goal~$\ESM$.
In the marriage case, the existence of a stable matching is always guaranteed, so the question is only interesting in the roommates setting.

\begin{restatable}[\linkproof{thm:CSR-bribery-ESM}]{theorem}{thmCSRbriberyESM}
\label{thm:CSR-bribery-ESM}
  \CSR-$\Reorder$-$\ESM$, \CSR-$\Reorder$-$\PSM$, \CSR-$\Swap$-$\ESM$, and \CSR-$\Swap$-$\PSM$ are all $\np$-complete.
\end{restatable}
\prooftoappendix{thm:CSR-bribery-ESM}{\thmCSRbriberyESM*}{
  Membership in~$\np$ follows as in \cref{thm:CSR-AddAg-ExistsSM}: it suffices to apply the operations and run Irving's algorithm~\cite{irv:j:efficient-algorithm-for-stable-roommates-problem} on the modified instance.

  For $\np$-hardness, we reduce from \textsc{Independent Set} (see \cref{thm:CSR-AddAg-ExistsSM} for the definition).
  The reduction we present adapts the construction of \cref{thm:CSR-AddAg-ExistsSM}: we keep the role of an agent of~$V$ (the addable agents in the original construction) as being matched to a selector, but make each agent $v \in V$ appear in a small ``twin gadget'' that forces a briber to reorder a specific private agent (rather than $v$ itself), thereby preserving the independent set argument.

  Let $(G=(V,E),k)$ be the input. We may assume without loss of generality that each vertex has degree at least~$2$.
    
  \smallskip
  \noindent\textbf{Construction.}
  For each vertex $v\in V$, introduce four \myemph{vertex-side agents}: $v$, its \myemph{twin}~$v'$, and two \myemph{private agents}~$a_v$ and~$b_v$.
  For each $j\in[k]$, introduce two \myemph{selector gadgets}, each containing $2k+3$ agents arranged cyclically:
  \[
    S_j=\{s_j^0,s_j^1,\dots,s_j^{2k+2}\}, \qquad T_j=\{t_j^0,t_j^1,\dots,t_j^{2k+2}\}.
  \]
  Let $N(v')\coloneqq\{u'\colon u\in N(v)\}$ be the twins corresponding to the neighbors of $v$.
  The preferences are
  \[
    \begin{array}{ll}
      v\colon a_v\succ N(v)\succ s_1^0\succ\dots\succ s_k^0, & \forall v\in V; \\[2pt]
      v'\colon a_v\succ b_v\succ N(v')\succ t_1^0\succ\dots\succ t_k^0,\qquad & \forall v\in V; \\[2pt]
      a_v\colon v\succ v'\succ b_v, & \forall v\in V; \\[2pt]
      b_v\colon a_v\succ v', & \forall v\in V; \\[2pt]
      s_j^0\colon V\succ s_j^1\succ s_j^2\succ\dots\succ s_j^{2k+2}, & \forall j\in[k]; \\[2pt]
      s_j^z\colon s_j^{z+1}\succ s_j^{z+2}\succ\dots\succ s_j^{z-1}, & \forall j\in[k],\,z\in[2k+2]; \\[2pt]
      t_j^0\colon V'\succ t_j^1\succ t_j^2\succ\dots\succ t_j^{2k+2}, & \forall j\in[k]; \\[2pt]
      t_j^z\colon t_j^{z+1}\succ t_j^{z+2}\succ\dots\succ t_j^{z-1}, & \forall j\in[k],\,z\in[2k+2],
    \end{array}
  \]
  where $V'\coloneqq\{v'\colon v\in V\}$, indices on the cyclic preferences are taken modulo $2k+3$, and each set is interpreted as an arbitrary but fixed strict linear order of its elements.
  The budget is $\ell\coloneqq k$ for the reorder action and $\ell\coloneqq 2k$ for the swap action.

  \begin{clm}\label{clm:csr-reorder-existssm-cyclegadget}
    Each cyclic gadget $S_j$ (and analogously~$T_j$) satisfies:
    \begin{compactenum}[(i)]
      \item\label{clm:cyclegadget1} if $s_j^0$ is matched outside of~$S_j$, then $\{\{s_j^i,s_j^{i+k+1}\}\colon i\in[k+1]\}$ (with the second index taken modulo~$2k+3$) is a stable matching on $S_j\setminus\{s_j^0\}$ that produces no blocking pair with $s_j^0$ either;
      \item\label{clm:cyclegadget2} if some stable matching of a modified instance leaves $s_j^0$ unmatched or matched within~$S_j$, then the modification includes at least $k+1$ reorder operations, respectively at least $\binom{k+2}{2}$ swap operations, acting on agents of~$S_j$; note that $\binom{k+2}{2}=\frac{(k+1)(k+2)}{2}>2k$.
    \end{compactenum}
  \end{clm}
  \begin{claimproof}
    For~(i), let $M_j\coloneqq\{\{s_j^i,s_j^{i+k+1}\}\colon i\in [k+1]\}$, where the second index is taken modulo~$2k+3$.
    For each $a\in[2k+2]$, let $d_a$ denote the cyclic forward distance from~$s_j^a$ to its partner in~$M_j$; from the construction,
    \[
      d_a=k+1\,\, \text{ if } a\in[k+1], \qquad d_a=k+2\,\, \text{ if } k+2\le a\le 2k+2.
    \]

    Consider any non-matched pair $\{s_j^a,s_j^b\}\subseteq S_j\setminus\{s_j^0\}$ with $a<b$.
    Writing $d_{ab}=b-a$ and $d_{ba}=2k+3-d_{ab}$ for the forward cyclic distances, $\{s_j^a,s_j^b\}$ blocks $M_j$ exactly when (a) $d_{ab}<d_a$ and (b) $d_{ba}<d_b$, which combine to $d_a+d_b>2k+3$.
    The only way this holds is $k+2\le a<b\le 2k+2$, giving $d_a+d_b=2k+4$. Then (a) means $d_{ab}<k+2$, and (b) means $d_{ba}=2k+3-d_{ab}<k+2$; these inequalities imply $k+1<d_{ab}<k+2$, a contradiction.
    Hence, $M_j$ admits no internal blocking pair.
    Moreover, $s_j^0$ ranks every agent in~$V$ above every agent of $S_j\setminus\{s_j^0\}$, so its external partner is preferred over any~$s_j^a$, ruling out any blocking pair of the form~$\{s_j^0,s_j^a\}$ as well.

    For~(ii), let $M$ be a stable matching of the modified instance that leaves $s_j^0$ unmatched or matched within~$S_j$.
    First, observe that at most one agent of~$S_j$ is unmatched by~$M$: all pairs of agents within~$S_j$ are acceptable, and two unmatched agents always form a blocking pair, irrespective of any preference modifications.
    Since $|S_j|=2k+3$ is odd, exactly one agent~$s\in S_j$ is unmatched, and $M$ contains exactly $k+1$ pairs within~$S_j$.
    Fix such a pair~$\{s_j^x,s_j^y\}$.
    The sets $F_x=\{s_j^{x+1},\dots,s_j^{y-1}\}$ and $F_y=\{s_j^{y+1},\dots,s_j^{x-1}\}$ (indices taken modulo~$2k+3$) partition $S_j\setminus\{s_j^x,s_j^y\}$, so the unmatched agent~$s$ lies in exactly one of them; by symmetry, we may assume $s \in F_x$.
    In the original preferences of~$s_j^x$, every agent of~$F_x$ precedes~$s_j^y$; since $s$ is unmatched, the pair $\{s_j^x,s\}$ blocks~$M$ unless the preferences of~$s_j^x$ have been modified so that its partner~$s_j^y$ precedes~$s$ (note that modifying the preferences of the unmatched agent~$s$ is of no use here).
    In the reorder setting, this requires one reorder operation on~$s_j^x$; since each agent of~$S_j$ lies in at most one pair of~$M$, we obtain $k+1$ pairwise-distinct reorders on agents of~$S_j$.

    In the swap setting, we count more carefully.
    Let $d(\cdot,\cdot)$ denote the forward cyclic distance within~$S_j$.
    Moving~$s_j^y$ ahead of~$s$ in the preference list of~$s_j^x$ requires at least $d(s,s_j^y)$ swaps on this list: the agents $s$ and~$s_j^y$ occupy positions $d(s_j^x,s)$ and ${d(s_j^x,s_j^y)=d(s_j^x,s)+d(s,s_j^y)}$ among the gadget agents in the list of~$s_j^x$, and to invert their relative order, the pair $(s,s_j^y)$ itself as well as, for each of the $d(s,s_j^y)-1$ agents~$z$ between them, at least one of the pairs $(s,z)$ and~$(z,s_j^y)$ must be inverted.
    Now observe that $s_j^y$ is precisely the endpoint of the pair~$\{s_j^x,s_j^y\}$ that is reached first when moving forward from~$s$ along the cycle.
    Since the $k+1$ pairs of~$M$ within~$S_j$ are pairwise disjoint, these endpoints are pairwise distinct agents, and hence the corresponding distances~$d(s,\cdot)$ are pairwise distinct positive integers.
    Summing over the $k+1$ pairs, the total number of swaps on agents of~$S_j$ is at least $1+2+\dots+(k+1)=\binom{k+2}{2}$.
  \end{claimproof}

  Let us show the correctness of our reduction.
  For the forward direction, assume that $G$ admits an independent set~$X\subseteq V$ of size~$k$, and fix two bijections $\pi_s,\pi_t\colon X\to[k]$ such that the matching $\{\{v,s^0_{\pi_s(v)}\}\colon v\in X\}$ is stable in the sub-instance induced by $X\cup\{s_1^0,\dots,s_k^0\}$ and the matching $\{\{v',t^0_{\pi_t(v)}\}\colon v\in X\}$ is stable in the sub-instance induced by $\{v'\colon v\in X\}\cup\{t_1^0,\dots,t_k^0\}$; such bijections exist, since both sub-instances are \textsc{Stable Marriage} instances with complete preference lists (see footnote~\ref{fn:stableway}).
  We perform one reorder per vertex of~$X$:
  for each $v\in X$, replace $a_v$'s list by $b_v\succ v\succ v'$.
  In the swap setting, this is realized by two swaps on~$a_v$'s list (moving $b_v$ from position~3 to position~1), giving $2k$ swaps in total.
  Define the matching
  \begin{align*}
    M\coloneqq{}&\{\{v,s^0_{\pi_s(v)}\},\{v',t^0_{\pi_t(v)}\},\{a_v,b_v\}\colon v\in X\} \\
    {}\cup{}&\{\{a_v,v\},\{b_v,v'\}\colon v\in V\setminus X\}\cup M_{\mathrm{sel}},
  \end{align*}
  where $M_{\mathrm{sel}}$ is the internal matching  from \cref{clm:csr-reorder-existssm-cyclegadget}\eqref{clm:cyclegadget1} within the selector gadgets, containing $\{s_j^i,s_j^{i+k+1}\}$ and $\{t_j^i,t_j^{i+k+1}\}$ for each $i\in[k+1]$ and~$j \in [k]$.
  The matching~$M$ covers every agent, and we verify that it is stable in the modified instance:
  \begin{compactitem}
    \item No blocking pair lies inside a selector gadget, by \cref{clm:csr-reorder-existssm-cyclegadget}\eqref{clm:cyclegadget1}.
    \item No blocking pair contains a distinguished gadget agent~$s_j^0$: a pair $\{v,s_j^0\}$ with $v\in X$ does not block due to the stability of the sub-matching $\{\{v,s^0_{\pi_s(v)}\}\colon v\in X\}$, and a pair $\{v,s_j^0\}$ with $v\in V\setminus X$ does not block since such a~$v$ is matched with its top choice~$a_v$.
    The analogous statement holds for the agents~$t_j^0$ and the twins~$v'$.
    \item For $v\in X$, agents $a_v,b_v$ are matched to each other with $a_v$'s reordered top choice~$b_v$, so $\{a_v,v\},\{a_v,v'\},\{a_v,b_v\}$ do not block.
    Moreover, $\{b_v,v'\}$ does not block since $b_v$'s top choice~$a_v$ is its partner.
    \item For $v\in X$, $\{v,u\}$ with $u\in N(v)$ does not block: since $X$ is an independent set, $u\notin X$, and thus $u$ is matched with~$a_u$, its top choice.
    An analogous argument shows that pairs of the form $\{v',u'\}$ do not block~$M$ either.
    \item For $v\in V\setminus X$, the pairs $\{a_v,v\},\{b_v,v'\}$ match both $v$ and~$a_v$ with their top choice, and match both $v'$ and~$b_v$ with their top choice excluding $a_v$; hence, none of these agents are in blocking pairs.
  \end{compactitem}
  \smallskip

  For the backward direction, assume that a set of operations within the budget yields an instance~$I'$ admitting a stable matching~$M$.
  By  \cref{clm:csr-reorder-existssm-cyclegadget}\eqref{clm:cyclegadget2}, any selector gadget whose distinguished agent ($s_j^0$ or $t_j^0$) is not matched externally absorbs at least $k+1$ reorder operations, respectively at least $\binom{k+2}{2}>2k$ swap operations, on its own---already exceeding the respective global budget.
  Hence every $s_j^0$ must be matched with an agent in~$V$ in~$M$, and every $t_j^0$ with an agent in~$V'$.
  Let $S^0=\{s_j^0\colon j \in [k]\}$ and $T^0=\{t_j^0\colon j \in [k]\}$, and 
  define the vertex sets $X_s\coloneqq\{v\in V\colon M(v) \in S^0\}$ and $X_t\coloneqq\{v\in V\colon M(v) \in T^0\}$; then we have $|X_s|=|X_t|=k$.

  \emph{Step 1: $X_s=X_t$.}
  For each $v\in V$, let $Q_v\coloneqq\{v,v',a_v,b_v\}$; note that these sets are pairwise disjoint.
  For each $v\in X_s$, the pair $\{a_v,v\}$ blocks~$M$ unless an operation acts on an agent of~$Q_v$: indeed, $a_v$ and~$v$ are each other's top choice, while $v$ is matched to an agent of~$S^0$.
  For each $v\in X_t$, some pair blocks~$M$ as well unless an operation acts on an agent of~$Q_v$:
  if $M(b_v)\neq a_v$, then $b_v$ is unmatched (its only acceptable partners are $a_v$ and~$v'$, and $v'$ is matched to an agent of~$T^0$), and $\{b_v,v'\}$ blocks unless the preferences of~$v'$ are modified---note that the preferences of the unmatched agent~$b_v$ are irrelevant for this pair;
  if $M(b_v)=a_v$, then $\{a_v,v'\}$ blocks unless the preferences of~$a_v$ or~$v'$ are modified.
  Hence, the number of operations is at least~$|X_s\cup X_t|$.
  For the $\Reorder$ action, the budget of~$k$ reorders thus yields $|X_s\cup X_t|\le k$, and $|X_s|=|X_t|=k$ implies $X_s=X_t$.

  For the $\Swap$ action, we show that each $v\in X_t$ in fact requires at least \emph{two} swaps on agents of~$Q_v$.
  Suppose for contradiction that at most one swap acts on agents of~$Q_v$.
  A single swap cannot move $M(v')\in T^0$ above~$b_v$ in $v'$'s list, since $b_v$ occupies the second position and $M(v')$ a position of at least~five (recall that $|N(v)|\ge 2$); thus $v'$ prefers both~$a_v$ and~$b_v$ to its partner.
  Consequently, $M(b_v)=a_v$ must hold, as otherwise $\{b_v,v'\}$ blocks~$M$.
  Moreover, $a_v$ must prefer~$b_v$ to~$v'$, as otherwise $\{a_v,v'\}$ blocks~$M$; this requires a swap on~$a_v$'s list---necessarily the unique one---exchanging~$v'$ and~$b_v$ and yielding the list $v\succ b_v\succ v'$.
  But then $a_v$ prefers~$v$ to its partner~$b_v$, and $v$, whose preferences are unchanged, prefers its top choice~$a_v$ to its partner; hence $\{a_v,v\}$ blocks~$M$, a contradiction.
  Together with at least one swap on agents of~$Q_v$ for each $v\in X_s\setminus X_t$, the budget of~$2k$ swaps yields
  $2|X_t|+|X_s\setminus X_t|=2k+|X_s\setminus X_t|\le 2k$; hence $X_s\subseteq X_t$, and $|X_s|=|X_t|=k$ implies $X_s=X_t$.
  Write $X\coloneqq X_s=X_t$.

  \emph{Step 2: each $v\in X$ is handled by reordering $a_v$.}
  For $v\in X$, both $\{a_v,v\}$ and $\{a_v,v'\}$ must be prevented from blocking (since $v$ is matched with an agent in~$S^0$ and $v'$ with an agent in~$T^0$, both lower in their original preferences than~$a_v$).
  Placing $b_v$ above both~$v$ and~$v'$ in $a_v$'s list (a single reorder operation, or two swaps) resolves all three potential blocks $\{a_v,v\},\{a_v,v'\},\{b_v,v'\}$ simultaneously (the last because $b_v$ is then $a_v$'s match and $b_v$ ranks $a_v$ first).

  Reordering~$v$ instead leaves $\{a_v,v'\}$ blocking, %
  requiring a second reorder; reordering~$v'$ symmetrically leaves $\{a_v,v\}$ blocking; reordering~$b_v$ alone cannot fix any of the three blocks.
  Hence each $v\in X$ that is not handled via reordering~$a_v$ costs at least two reorder operations;
  for the $\Swap$ action, the per-vertex budget of two swaps can move neither $M(v)$ above~$a_v$ in $v$'s list nor $M(v')$ above~$b_v$ in $v'$'s list (each requires at least three swaps); hence both $\{a_v,v\}$ and $\{a_v,v'\}$ must be resolved via $a_v$'s list, which requires both available swaps: moving~$b_v$ above~$v$ and~$v'$.
  Since the budget equals $k=|X|$ reorder operations (resp., $2k$ swaps) and each $v\in X$ requires at least one reorder (resp., two swaps) on agents of~$Q_v$ by Step~1, we conclude that for each $v\in X$, exactly the preference list of~$a_v$ is modified, and no other preference list is modified at all.

  \emph{Step 3: $X$ is an independent set.}
  As we have shown in Step~2, each $v\in V$ retains its original preference list: only the preferences of the agents~$a_v$ for $v\in X$ have been modified.
  Suppose for contradiction that $\{u,v\}\in E$ with $u,v\in X$.
  In~$M$, both~$u$ and~$v$ are matched to some agent in a selector gadget. %
  Since $u\in N(v)$ and~$v$'s preferences place $N(v)$ above all agents in selector gadgets, $v$ prefers~$u$ to its partner~$M(v)$; symmetrically, $u$ prefers $v$ to its partner~$M(u)$.
  Hence, $\{u,v\}$ blocks~$M$, a contradiction.
  Therefore $X$ is an independent set of size~$k$ in~$G$.

  Finally, hardness for~$\PSM$ follows because the matching $M$ constructed in the forward direction covers every agent, and---by the well-known fact that all stable matchings in a \SR\ instance cover the same set of agents~\cite[Theorem~4.5.2]{gus-irv:b:the-stable-marriage-problem}---every stable matching in the modified instance is also perfect.
}

\begin{remark}
  Hardness for the swap action with the $\PSM$ goal is in fact already known in the strictly easier marriage setting: in the context of the trade-off between optimality (size of the matching) and stability, \CSM-$\Swap$-$\PSM$ has been shown to be $\mathrm{W}[1]$-hard parameterized by the swap budget~$\ell$~\cite{che-sko-sor:j:tradeoff}, and this hardness lifts to \CSR-$\Swap$-$\PSM$ since \textsc{Stable Marriage} is a special case of \textsc{Stable Roommates}.
  The argument above is nonetheless included because it yields the $\ESM$ variants and the reorder cases in the same construction.
\end{remark}

\section{Conclusion}\label{sec:conclusion}
We studied \emph{control} and \emph{bribery} for stable matchings under five manipulative actions and six goals, mapping the complete complexity landscape for both the marriage and the roommates settings (\cref{tab:summary-for-SM}).

Regarding our $\np$-hardness results, it is important to keep in mind that our results do not necessarily offer a barrier against control in practice.
Nevertheless, 
several of our reductions go beyond classical $\np$-hardness: those in \cref{thm:CSM-AddAg-MA,cor:CSR-AddAg-MA,thm:CSR-AddAg-ExistsSM,thm:CSM-CSR-Reorder-MA,prop:CSM-CSR-Swap-MA,thm:CSR-bribery-ESM} reduce from \textsc{Clique} or \textsc{Independent Set} with budget~$\ell$ polynomial in the solution size,
yielding $\mathrm{W}[1]$-hardness parameterized by~$\ell$.
The remaining new hardness results do not yield such parameterized lower bounds: the reductions in \cref{thm:CSR-AddAg-SM,thm:CSR-DelAcc-USM,thm:CSR-Swap-MS} have budgets that are not bounded by the source parameter, and \cref{prop:CSR-Reorder-MS} reduces from \textsc{Vertex Cover}, which is fixed-parameter tractable; hence, the parameterized complexity of these cases remains open.

Several directions remain open.
First, since several of our intractable cases are already $\mathrm{W}[1]$-hard parameterized by the number $\ell$ of applied actions, a natural next step is to identify alternative parameters that may yield fixed-parameter tractability, such as bounds on preference-list lengths, structural restrictions of the acceptability graph (e.g., bounded treewidth/vertex cover), distances to restricted domains, or combined parameters (e.g., $\ell$ together with such a structural bound); the parameterized complexity of the cases that our reductions leave open is itself a concrete target.
Second, one can study restricted preference classes (e.g., single-peaked/single-crossing preferences~\cite{bal-har:j:characterization-single-peaked,BreCheFinNie2020-spscSM-jaamas,BCW12}) that may restore tractability for some action--goal combinations.
Finally, it would be interesting to extend the same action/goal palette to richer matching models (ties, many-to-one, matching with contracts, distributional constraints) and to compare systematically with emerging control landscapes in coalition formation settings.

\subsection*{Acknowledgment}
We are grateful to J\"org Rothe for enjoyable and fruitful discussions, and for his invitation to prepare a survey on control problems in the area of matching under preferences;  the work contained in this paper originated in our efforts to write that survey. 
J.~Chen is supported by the Vienna Science and Technology Fund (WWTF)~[10.47379/VRG18012].
I.~Schlotter is supported by the Hungarian Academy of Sciences under its Momentum Programme (LP2021-2) and its J\'anos Bolyai Research Scholarship. 

{
\small
\bibliographystyle{splncs04}
\bibliography{control}

\begin{thebibliography}{10}
\providecommand{\url}[1]{\texttt{#1}}
\providecommand{\urlprefix}{URL }
\providecommand{\doi}[1]{https://doi.org/#1}

\bibitem{abr-bir-man:c:almost-stable-roommates}
Abraham, D.J., Bir\'o, P., Manlove, D.F.: ``{A}lmost stable'' matchings in the
  roommates problem. In: Proc. of the 3rd Workshop on Approximation and Online
  Algorithms (WAOA~'05). LNCS, vol.~3879, pp. 1--14 (2006)

\bibitem{bal-har:j:characterization-single-peaked}
Ballester, M., Haeringer, G.: A characterization of the single-peaked domain.
  Social Choice and Welfare  \textbf{36}(2),  305--322 (2011)

\bibitem{bar-tov-tri:j:control}
{Bartholdi III}, J., Tovey, C., Trick, M.: How hard is it to control an
  election? Mathematical and Computer Modelling  \textbf{16}(8/9),  27--40
  (1992)

\bibitem{bar-tov-tri:j:manipulating}
{Bartholdi~III}, J.J., Tovey, C.A., Trick, M.A.: The computational difficulty
  of manipulating an election. Social Choice and Welfare  \textbf{6}(3),
  227--241 (1989)

\bibitem{ber-csa-kir:j:manipulation-SR}
B\'erczi, K., Cs\'aji, G., Kir\'aly, T.: Manipulating the outcome of stable
  marriage and roommates problems. Games and Economic Behavior  \textbf{147},
  407--428 (2024)

\bibitem{BermanKarpinskiScott2003}
Berman, P., Karpinski, M., Scott, A.D.: Approximation hardness of short
  symmetric instances of {MAX-3SAT}. Tech. rep. (2003), report no. 49

\bibitem{bir-man-mit:j:size-vs-stability}
Bir\'o, P., Manlove, D.F., Mittal, S.: Size versus stability in the marriage
  problem. Theoretical Computer Science  \textbf{411}(16--18),  1828--1841
  (2010)

\bibitem{boe-bre-hee-nie:j:control-bribery-in-SM}
Boehmer, N., Bredereck, R., Heeger, K., Niedermeier, R.: Bribery and control in
  stable marriage. Journal of Artificial Intelligence Research  \textbf{71},
  993--1048 (2021)

\bibitem{BreCheFinNie2020-spscSM-jaamas}
Bredereck, R., Chen, J., Finnendahl, U.P., Niedermeier, R.: Stable roommate
  with narcissistic, single-peaked, and single-crossing preferences. Autonomous
  Agents and Multi-Agent Systems  \textbf{34}(53),  1--29 (2020)

\bibitem{BCW12}
Bredereck, R., Chen, J., Woeginger, G.J.: A characterization of the
  single-crossing domain. Social Choice and Welfare  \textbf{41}(4),  989--998
  (2013)

\bibitem{chen-gut-mus-sim:aamas26:control-hedonic}
Chen, J., Guttmann, J., Mustajba{\v{s}}i{\'c}, M., Simola, S.: Control in
  hedonic games. In: Proc. of the 25th International Conference on Autonomous
  Agents and Multiagent Systems (AAMAS~'26). pp. 920--928 (2026)

\bibitem{che-herm-sor-yed:c:-par-stable}
Chen, J., Hermelin, D., Sorge, M., Yedidsion, H.: How hard is it to satisfy
  (almost) all roommates? In: Proc. of the 45th International Colloquium on
  Automata, Languages, and Programming (ICALP '18). LIPIcs, vol.~107, pp.
  35:1--35:15. Schloss Dagstuhl -- Leibniz-Zentrum f{\"u}r Informatik (2018)

\bibitem{ChenKNRSS-ijcai25}
Chen, J., Kaczmarek, J., Nüsken, P., Rothe, J., Schlotter, I., Seeger, T.:
  Control in computational social choice. In: Proc. of the 34th International
  Joint Conference on Artificial Intelligence (IJCAI~'25). pp. 10391--10399
  (2025), {S}urvey {T}rack.

\bibitem{che-sko-sor:j:tradeoff}
Chen, J., Skowron, P., Sorge, M.: Matchings under preferences: Strength of
  stability and tradeoffs. ACM Trans. Econ. Comput.  \textbf{9}(4) (2021)

\bibitem{che-ros:x:optimal-SM}
Cheng, C.T., Rosenbaum, W.: Bridging the gap between stable marriage and stable
  roommates: A parameterized algorithm for optimal stable matchings (2026),
  arXiv:2603.26943 [cs.DS]

\bibitem{for-hop-wyl:j:back-and-forth-paths}
Fortune, S., Hopcroft, J., Wyllie, J.: The directed subgraph homeomorphism
  problem. Theoretical Computer Science  \textbf{10}(2),  111--121 (1980)

\bibitem{gal-sha:j:college-admissions-and-the-stability-of-marriage}
Gale, D., Shapley, L.S.: College admissions and the stability of marriage. The
  American Mathematical Monthly  \textbf{69}(1),  9--15 (1962)

\bibitem{gar-joh:b:int}
Garey, M.R., Johnson, D.S.: Computers and Intractability: {A} Guide to the
  Theory of {NP}-Completeness. {W. H. Freeman and Company} (1979)

\bibitem{gup-jai:c:manipulation-matching}
Gupta, S., Jain, P.: Manipulation with(out) money in matching market. In: Proc.
  of the 8th International Conference on Algorithmic Decision Theory (ADT '24).
  LNCS, vol. 15248, p. 273–287 (2025)

\bibitem{gus-irv:b:the-stable-marriage-problem}
Gusfield, D., Irving, R.W.: The stable marriage problem: {S}tructure and
  algorithms. Foundations of Computing Series, MIT Press, Cambridge, MA (1989)

\bibitem{irv:j:efficient-algorithm-for-stable-roommates-problem}
Irving, R.W.: An efficient algorithm for the “stable roommates” problem.
  Journal of Algorithms  \textbf{6}(4),  577--595 (1985)

\bibitem{irv-lea-gus:j:efficient-algorithm-for-optimal-stable-marriage}
Irving, R.W., Leather, P., Gusfield, D.: An efficient algorithm for the
  ``optimal'' stable marriage. Journal of the ACM  \textbf{34}(3),  532--543
  (1987)

\bibitem{kam:j:super-stable-agent-deletion}
Kamiyama, N.: Modifying an instance of the super-stable matching problem.
  Information Processing Letters  \textbf{189},  106549 (2025)

\bibitem{man:b:algorithmics-of-matching-under-preferences}
Manlove, D.F.: Algorithmics of Matching Under Preferences, Series on
  Theoretical Computer Science, vol.~2. World Scientific Publishing (2013)

\bibitem{mni-sch:j:SM-covering-constraints}
Mnich, M., Schlotter, I.: Stable marriage with covering constraints: A complete
  computational trichotomy. Algorithmica  \textbf{82}(1),  1136--1188 (2020)

\bibitem{tan:j:maximum-stable-matching}
Tan, J.J.: A maximum stable matching for the roommates problem. {BIT} Numerical
  Mathematics  \textbf{30},  631--640 (1990)

\bibitem{tan:j:stable-roommates}
Tan, J.J.: A necessary and sufficient condition for the existence of a complete
  stable matching. Journal of Algorithms  \textbf{12},  154--178 (1991)

\bibitem{tan:j:stable-partitions}
Tan, J.J.: Stable matchings and stable partitions. International Journal of
  Computer Mathematics  \textbf{39},  11--20 (1991)

\bibitem{tan-hsu:j:-generalization-of-the-stable-matching-problem}
Tan, J.J., Hsueh, Y.C.: A generalization of the stable matching problem.
  Discrete Applied Mathematics  \textbf{59},  87--102 (1995)

\end{thebibliography}
}

\ifshort
\else
\begin{subappendices}
\renewcommand{\thesection}{\Alph{section}}%
\crefalias{section}{appendix}
\crefalias{subsection}{appendix}
\crefalias{subsubsection}{appendix}

\clearpage
\appendixtext

\end{subappendices}

\fi 

\end{document}